\documentclass[aip, jmp, amsmath,amssymb,preprint, author-numerical, longbibliography]{revtex4-1}

\usepackage{graphicx}
\usepackage{dcolumn}
\usepackage{bm}
\usepackage{amsthm}

\def\intprod{\vrule height 0pt depth 0.4pt width 3pt \vrule height 7pt depth 0.4pt
\kern 3pt}

\def\prodint{\kern 2pt \vrule height 7pt depth 0.4pt  \vrule %
height 0pt depth 0.4pt width 3pt %
\kern 2pt}

\newcommand{\bfi}{\bfseries\itshape}

\newtheorem{theorem}{Theorem}[section]
\newtheorem{corollary}[theorem]{Corollary}
\newtheorem{proposition}[theorem]{Proposition}

\newtheorem{lemma}[theorem]{Lemma}
\newtheorem{definition}[theorem]{Definition}

\begin{document}


\title[The Hamilton-Pontryagin Principle and Multi-Dirac Structures]{%
The Hamilton-Pontryagin Principle and Multi-Dirac Structures \\ 
for Classical Field Theories\footnote{We dedicate this paper to the memory of our long-time mentor, collaborator and friend Jerry Marsden.}}

\author{J. Vankerschaver}
\altaffiliation[Also at ]{Department of Mathematics, Ghent University, Krijgslaan 281, B-9000 Ghent, Belgium.}
\email{jvankers@math.ucsd.edu.}
\affiliation{Department of Mathematics, University of California at San Diego, 9500 Gilman Drive, San Diego CA 92093, USA.}

\author{H. Yoshimura}%
 \email{yoshimura@waseda.jp.}
\affiliation{Applied Mechanics and Aerospace Engineering, Waseda University, Okubo, Shinjuku, Tokyo, 169-8555, Japan.}

\author{M. Leok}
\email{mleok@math.ucsd.edu.}
\affiliation{Department of Mathematics, University of California at San Diego, 9500 Gilman Drive, San Diego CA 92093, USA.}

\date{\today}


\begin{abstract}
We introduce a variational principle for field theories, referred to as the Hamilton-Pontryagin principle, and we show that the resulting field equations are the Euler-Lagrange equations in implicit form.  Secondly, we introduce multi-Dirac structures as a graded analog of standard Dirac structures, and we show that the graph of a multisymplectic form determines a multi-Dirac structure.  We then discuss the role of multi-Dirac structures in field theory by showing that the implicit Euler-Lagrange equations for fields obtained from the Hamilton-Pontryagin principle can be described intrinsically using multi-Dirac structures.   Lastly, we show a number of illustrative examples, including time-dependent mechanics, nonlinear scalar fields, Maxwell's equations, and elastostatics.
\end{abstract}

\pacs{03.50.-z, 02.30.Xx, 45.20.Jj}
\keywords{Variational principles, classical field theories, multisymplectic formulation, Hamilton-Pontryagin principle, multi-Dirac structures}

\maketitle


\section{Introduction}
 
In this paper,  we describe classical field theories by means of a new variational principle, referred to as the {\it Hamilton-Pontryagin variational principle} and we show that the resulting field equations can be described in terms of {\it multi-Dirac structures}, a field theoretic extension of the concept of Dirac structures introduced by Courant in Ref.~\onlinecite{Co1990}.   To set the stage for the paper, we begin with reviewing the Hamilton-Pontryagin principle and Dirac structures in the context of mechanical systems.

\subsection{The Hamilton-Pontryagin Variational Principle in Mechanics}

Let $L(q^{i}, v^{i})$ be a Lagrangian.   The {\it Hamilton-Pontryagin principle} is a variational principle in which the position coordinates $q^{i}$ and the velocity coordinates $v^{i}$ are treated independently and the relation $\dot{q}^{i} = v^{i}$ is imposed as a constraint by means of a Lagrange multiplier $p_{i}$.  This leads to an action functional of the form 
\begin{equation} \label{HPmech}
\begin{split}
	S(q, v, p) &= \int_{t_0}^{t_1} 
	\big(L(q^{i}(t), v^{i}(t)) + \left< p_{i}(t), \dot{q}^{i}(t) - v^{i}(t) \right>\big) dt\\
	&{=\int_{t_0}^{t_1} \big( \left< p_{i}(t), \dot{q}^{i}(t) \right> -E(q^{i}(t), v^{i}(t), p_{i}(t)) \big)dt,}
\end{split}  
\end{equation}
{where $E(q^{i},v^{i},p_{i}):= \left< p_{i}, v^{i} \right>-L(q^{i},v^{i})$ is the generalized energy associated to the Lagrangian.
By taking arbitrary variations with respect to $q^{i}$, $v^{i}$ and $p_{i}$, we obtain the Euler-Lagrange equations in implicit form: 
\begin{equation} \label{hpeqns}
	p_i = \frac{\partial L}{\partial v^i}, \quad 
	\dot{p}_i = \frac{\partial L}{\partial q^i}, 
	\quad \text{and} \quad \dot{q}^i = v^i.
\end{equation}

Variational principles in which the position and velocity coordinates are varied independently have a long history: the action functional \eqref{HPmech} in particular first appears in the work of Livens in Ref.~\onlinecite{Li1919} and was rediscovered in Refs.~\onlinecite{ScDeMiTs1998, GiTy1990}, among others.
The name ``Hamilton-Pontryagin principle'' was coined by Yoshimura and Marsden in Ref.~\onlinecite{YoMa2006b} because of the similarity with the Pontryagin maximum principle in optimal control theory.

It was shown in Ref.~\onlinecite{YoMa2006b} that the Hamilton-Pontryagin principle and the associated implicit equations of motion can be described in an intrinsic way by means of Dirac structures.  A  discrete version of the Hamilton-Pontryagin principle was developed in Refs.~\onlinecite{BoMa2009, LeOh2011} with a view towards the design of accurate variational integrators.

\subsection{Dirac Structures and Lagrange-Dirac Systems in Mechanics}

The notion of Dirac structures was originally developed by Courant and Weinstein in Refs.~\onlinecite{CoWe1988, Co1990} and Dorfman in Ref.~\onlinecite{Do1993} as a unification of (pre-)symplectic and Poisson structures.  It was soon realized that Dirac structures play an important role in mechanics.  In particular, it was shown that interconnected systems, such as LC circuits, and nonholonomic systems can be effectively formulated in the context of implicit Hamiltonian systems (Refs.~\onlinecite{Co1990, Do1993, CoWe1988, ScMa1995}).  On the Lagrangian side, it was shown in Ref.~\onlinecite{YoMa2006a} that Dirac structures induced from distributions on configuration manifolds naturally yield a notion of {\it implicit Lagrangian systems}, allowing for the description of mechanical systems with degenerate Lagrangians and nontrivial constraint distributions.  

{I}n order to show the link with the equations \eqref{hpeqns} obtained from the Hamilton-Pontryagin principle, let {$\Delta_{Q}$} be a distribution on a manifold $Q$.  Consider the canonical symplectic form $\Omega$ on $T^\ast Q$ and define $\Omega_{M}$ as the pullback of $\Omega$ to the {\it Pontryagin bundle} $M := TQ \oplus T^\ast Q$.  The set $D$, given by 
\[
	D = \{ (X, \alpha) \in TM \oplus T^\ast M \mid  \mathbf{i}_X {\Omega_{M}}  - \alpha \in \Delta_M^\circ, \quad X \in \Delta_M \},
\]
where $\Delta_M := T\pi_{Q, M}^{-1}(\Delta_{Q})$ (with {$\pi_{Q, M}: M \rightarrow Q$ is the Pontryagin bundle projection}), is a Dirac structure (see {Refs.~\onlinecite{YoMa2006a, YoMa2006b}}){; namely,} $D$ is maximally isotropic with respect to the following pairing on $TM \oplus T^\ast M$: 
\begin{equation} \label{mechpairing}
	\left<\left< (X, \alpha),(Y, \beta)\right>\right>_{-}
	= \frac{1}{2} ( \mathbf{i}_X \beta + \mathbf{i}_Y \alpha).
\end{equation}

The {\bfi implicit Euler-Lagrange equations} are given by $(X, \mathbf{d} E) \in D$.  It can be shown that these equations are given in coordinates by 
\begin{equation} \label{implhol}
	p_i  = \frac{\partial L}{\partial v^i}, 
	\quad 
	\dot{p}_i = \frac{\partial L}{\partial q^i} + \lambda_\alpha A^\alpha_i(q), 
	\quad
	\dot{q}^i = v^i, \quad \text{and} \quad
	A^\alpha_i(q) v^i = 0;
\end{equation}
see Ref.~\onlinecite{CeMaRaYo2010}.  Here, the forms $A^\alpha_i(q) d q^i$ form a basis for the annihilator $\Delta^\circ$, and the $\lambda_\alpha$ are Lagrange multipliers.  The equations \eqref{implhol} are readily seen to be the nonholonomic Euler-Lagrange equations in implicit form, and when the constraints vanish, $A^\alpha_i = 0$, we recover the implicit Euler--Lagrange equations \eqref{hpeqns}.

\subsection{The Hamilton-Pontryagin Principle for Field Theories}

For classical field theories, the Lagrangian $L(x^\mu, y^A, v^A_\mu)$ depends not only on the spacetime coordinates $x^\mu$ and the fields $y^A$, but also on the spacetime velocities $v^A_\mu$, which represent the derivatives of the fields $y^A$ with respect to $x^\mu$.  By comparison with \eqref{HPmech}, a candidate Hamilton-Pontryagin principle may then be given in coordinates by 
\begin{equation} \label{intro:fieldHP}
	S(y^A, y^A_\mu, p_A^\mu) = \int_U \left[
		p_A^\mu \left( \frac{\partial y^A}{\partial x^\mu} - v^A_\mu \right) + L(x^\mu, y^A, v^A_\mu) \right] d^{n+1} x.
\end{equation}
By taking variations with respect to $y^A$, $v^A_\mu$ and $p_A^\mu$ (where $\delta y^A$ is required to vanish on the boundary $\partial U$ of the integration domain), we obtain the following implicit equations: 
\begin{equation} \label{intro:impFE}
	\frac{\partial p_A^\mu}{\partial x^\mu} = \frac{\partial L}{\partial y^A}, 
	\quad
	p_A^\mu = \frac{\partial L}{\partial v^A_\mu}, 
	\quad \text{and} \quad 
	\frac{\partial y^A}{\partial x^\mu} = v^A_\mu,
\end{equation}
which clearly generalizes \eqref{hpeqns} to the case of field theories. In section~\ref{sec:hpprinciple}, we provide an intrinsic version of this variational principle, using the geometry of the first jet bundle and its dual as our starting point.  In contrast to mechanics, the jet bundle is an affine bundle and this makes the definition of the duality pairing used implicitly in \eqref{intro:fieldHP} somewhat more complicated.

\subsection{Multi-Dirac Structures}  

In section~\ref{sec:multidirac}, we describe the implicit field equations \eqref{intro:impFE} in terms of multi-Dirac structures.  These geometric structures were introduced in Ref.~\onlinecite{VaYoLe2011} as a graded version of the standard concept of Dirac structures, and their relevance for field theory lies in the fact that they generalize the concept of the ``graph of a multisymplectic structure.''
 
To make this more precise, we recall that a multisymplectic structure is a form $\Omega$ of degree at least two, which is closed and non-degenerate.  When the degree of $\Omega$ is two, $\Omega$ is a symplectic form, while if $\Omega$ is a top form, then it is necessarily a volume form.  The literature on multisymplectic field theories is by now very extensive, but for fundamental aspects we refer to Refs.~\onlinecite{GoSt1973,KiTu1979, BiSnFi1988,CaCrIb1991, GoIsMa1997, Br1997} and the references therein. 
 
Just as the graph of a symplectic form is a Dirac structure, the graph of a multisymplectic form $\Omega$ on a manifold $M$ turns out to be an example of what we have called a multi-Dirac structure in Ref.~\onlinecite{VaYoLe2011}.  Here, the graph of $\Omega$ is defined as follows.  Let $k+1$ be the degree of $\Omega$ and consider the mapping whereby an $l$-multivector field $\mathcal{X}_l$ (with $1 \le l \le k$) is contracted with $\Omega$: 
\[
	\mathcal{X}_l \in \mbox{$\bigwedge$}^l(TM) \mapsto \mathbf{i}_{\mathcal{X}_l} \Omega \in \mbox{$\bigwedge$}^{k+1 - l}(T^\ast M).
\]
The graph of this mapping determines a submanifold $D_l$ of $\bigwedge^l(TM) \times_M \bigwedge^{k+1-l}(T^\ast M)$, and by considering the direct sum 
\begin{equation} \label{outline_DS}
	D = D_1 \oplus \cdots \oplus D_{k} 
\end{equation}
of all such subbundles, we obtain a \emph{multi-Dirac structure}.  This structure is maximally isotropic under a graded antisymmetric version of the pairing \eqref{mechpairing}, and multi-Dirac structures can be considered more generally as \emph{graded versions of standard Dirac structures}. 

The link between multi-Dirac structures and the implicit field equations is provided by theorem~\ref{thm:implicitsympl}, wherein we show that the field equations obtained from the Hamilton-Pontryagin principle can be written as $\mathbf{i}_{\mathcal{X}} \Omega_M = (-1)^{n+2} \mathbf{d} E$.   In multi-Dirac form this becomes 
\begin{equation} \label{impliciteqns}
	(\mathcal{X}, (-1)^{n+2} \mathbf{d} E) \in D_{n+1}, 
\end{equation}
where $D_{n+1}$ is the component of the highest degree in \eqref{outline_DS}.  More generally, we arrive in definition~\ref{def:lagrangedirac} at the concept of a \emph{Lagrange-Dirac field theory}, which is a triple $(\mathcal{X}, E, D_{n+1})$ satisfying \eqref{impliciteqns}, where $D_{n + 1}$ belongs to a multi-Dirac structure $D$ which does not necessarily come from a multisymplectic form.

\subsection{Outline of the Paper}  

After introducing the Hamilton-Pontryagin principle for field theories in section~\ref{sec:hpprinciple}, we discuss the link with multi-Dirac structures in section~\ref{sec:multidirac}.  We finish the paper in section~\ref{sec:examples} with a number of examples of implicit field theories and multi-Dirac structures, taken from mechanics, electromagnetism and elasticity.  Among others, we show that the Hu-Washizu principle from linear elastostatics is a particular example of the Hamilton-Pontryagin principle.

The focus throughout this paper is on the Hamilton-Pontryagin principle and the application of multi-Dirac structures to field theories.  In the companion paper Ref.~\onlinecite{VaYoLe2011}, we  develop the algebraic and geometric properties of multi-Dirac structures.


\section{The Geometry of Jet Bundles} \label{sec:geometry}

In this section, we provide a quick overview of the geometry of jet bundles for the treatment of Lagrangian field theories.  Most of the material in this section is standard, and can be found in Refs.~\onlinecite{Sa1989, GoIsMa1997} and the references therein.

\subsection{Jet Bundles}

Let $X$ be an oriented manifold with volume form $\eta$, which in many examples is spacetime, and let $\pi_{XY}: Y \to X$ be a finite-dimensional fiber bundle which we call the {\bfi covariant configuration bundle}. The physical fields are sections of this bundle, which is the covariant analogue of the configuration space in classical mechanics. For future reference, we assume that the dimension of $X$ is $n + 1$ and that $\pi_{XY}$ is a bundle of rank $N$, so that $\dim Y = n + N + 1$.  Coordinates on $X$ are denoted $x^{\mu},\, \mu=1,2,...,n+1$, and fiber coordinates $Y$ are denoted by $y^{A}, \, A=1,...,N$ so that a section $\phi: X \rightarrow Y$ of $\pi_{XY}$ has coordinate representation $\phi(x)=(x^{\mu},y^{A}(x))$.   We will also assume that $X$ is equipped with a fixed volume form $\eta$ given in adapted local coordinates by $\eta = d^{n + 1} x := dx^1 \wedge \cdots dx^{n+1}$. 

The analogue in classical field theory of the tangent bundle in mechanics is the {\bfi first jet bundle} $J^1 Y$, which consists of equivalence classes of local sections of $\pi_{XY}$, where we say that two local sections $\phi_{1}$, $\phi_{2}$ of $Y$ are equivalent at $x \in X$ if their Taylor expansions around $x$ agree to the first order.  In other words, $\phi_1$ and $\phi_2$ are equivalent if $\phi_1(x) = \phi_2(x)$ and $T_x \phi_1 = T_x \phi_2$.  It follows that an equivalence class $[\phi]$ of local sections can be identified with a linear map $\gamma : T_x X \rightarrow T_y Y$ such that $T \pi_{XY} \circ \gamma = \mathrm{Id}_{T_x X}$.  As a consequence, $J^1 Y$ is a fiber bundle over $Y$, where the projection ${\pi_{Y,J^1 Y}}: J^1 Y \rightarrow Y$ is defined as follows: let $\gamma: T_x X \rightarrow T_y Y$ be an element of $J^1 Y$, then ${\pi_{Y,J^1 Y}}(\gamma) := y$.  Coordinates on $J^{1}Y$ are denoted by $(x^{\mu},y^{A},v^{A}_{\mu})$, where the fiber coordinates $v^{A}_{\mu}$ represent the first-order derivatives of a section.  They are defined by noting that any $\gamma \in J^1 Y$ is locally of the form
\begin{equation} \label{jetmap}
	\gamma =d x^\mu \otimes \left( \frac{\partial}{\partial x^\mu} + v^A_\mu \frac{\partial}{\partial y^A} \right).
\end{equation}
The first jet bundle $J^1 Y$ is an \emph{affine} bundle over $Y$, with underlying vector bundle $L(TX, VY)$ of linear maps from the tangent space $TX$ into the vertical bundle $VY$, defined as  
\[
V_{y}Y=\{ v\in T_{y}Y \mid T\pi_{XY}(v)=0 \}, \quad \text{for $y \in Y$}.
\]

Given any section $\phi: X \to Y$ of $\pi_{XY}$, its tangent map $T_{x}\phi$ at $x \in X$ is an element of $J^{1}_yY$, where $y = {\phi(x)}$.   Thus, the map $x \mapsto T_{x}\phi$ defines a section of  $J^{1}Y$, where now $J^1 Y$ is regarded as a bundle over $X$. This section is denoted $j^{1}\phi$ and is called the {\bfi first jet prolongation} of $\phi$. In coordinates, $j^{1}\phi$ is given by 
\[
x^{\mu} \mapsto (x^{\mu},y^{A}(x^{\mu}),\partial_{\nu}y^{A}(x^{\mu})),
\]
where $\partial_{\nu}=\partial/\partial{x^{\nu}}$. A section of the bundle $J^{1}Y \to X$ which is the first jet prolongation of a section $\phi: X \to Y$ is said to be {\bfi holonomic}.

\subsection{Dual Jet Bundles} 

Next, we consider the field-theoretic analogue of the cotangent bundle. We define the {\bfi dual jet bundle} $J^{1}Y^{\star}$ to be the vector bundle over $Y$ whose fiber at $y \in Y_{x}$ is the set of affine maps from $J^{1}_{y}Y$ to $\Lambda^{n+1}_{x}X$, where $\Lambda^{n+1} X$ denotes the bundle of $(n+1)$-forms on $X$. Note that since the space of affine maps from an affine space into a vector space forms a vector space, $J^{1}Y^{\star}$ is a vector bundle despite the fact that $J^{1}Y$ is only an affine bundle. A smooth section of $J^{1}Y^{\star}$ is an affine bundle map of $J^{1}Y$ to $\Lambda^{n+1}X$.  Any affine map from $J^{1}_{y}Y$ to $\Lambda^{n+1}_{x}X$ can locally be written as 
\[
v^{A}_{\mu} \mapsto (p+p_{A}^{\mu}v^{A}_{\mu})\,d^{n+1}x,
\]
where $d^{n+1}x := dx^{1}\wedge dx^{2} \wedge \cdots \wedge dx^{n+1}$ is a coordinate representation of the volume form $\eta$, so that coordinates on $J^{1}Y^{\star}$ are given by $(x^\mu, y^A,p_{A}^{\mu}, p)$.

Throughout this paper we will employ another useful description of $J^{1}Y^{\star}$. Consider again the bundle $\Lambda^{n+1}Y$ of $(n+1)$-forms on $Y$ and let $Z \subset \Lambda^{n+1}Y$ be the subbundle whose fiber over $y \in Y$ is given by
\[
Z_{y}=\{z \in \Lambda_{y}^{n+1} Y \mid \mathbf{i}_{v}\mathbf{i}_{w}z=0 \;\text{for all}\; v,w \in V_{y}Y \},
\]
where $\mathbf{i}_{v}$ denotes left interior multiplication by $v$.  In other words, the elements of $Z$ vanish when contracted with two or more vertical vectors.  The bundle $Z$ is canonically isomorphic to $J^1 Y^\star$ as a vector bundle over $Y$; this can be easily understood from the fact that the elements of $Z$ can locally be written as 
\[
z=p\,d^{n+1}x +p_{A}^{\mu}dy^{A} \wedge d^{n}x_{\mu},
\]
where $d^{n}x_{\mu} :=\partial_{\mu}$ {\Large$\lrcorner$} $d^{n+1}x$.

From now on, we will silently identify $J^1 Y^\star$ with $Z$.  The duality pairing between elements $\gamma \in J^1_{y_x} Y$ and $z \in Z_{y_x}$ can then be written as 
\[
	\left< z, \gamma \right> = \gamma^\ast z \in \Lambda_x^{n+1} X.
\]
In coordinates, $\left< z, \gamma \right> = (p_A^\mu v^A_\mu + p ) d^{n+1} x$.

\subsection{Canonical Multisymplectic Forms}

Analogous to the canonical symplectic forms on a cotangent bundle, there are canonical forms on $J^{1}Y^{\star}$.  
Let us first define the {\bfi canonical $(n+1)$-form} $\Theta$ on $J^1 Y^\star \cong Z$ by
\begin{equation} \label{canonnform}
\begin{split}
\Theta(z)(u_{1},...,u_{n+1})&=z(T\pi_{Y, J^1Y^\star}( u_{1}),...,T\pi_{Y, J^1Y^\star} (u_{n+1}))\\
&=(\pi_{Y, J^1Y^\star}^{\ast}z)(u_{1},...,u_{n+1}),
\end{split}
\end{equation}
where we have interpreted $z \in Z \cong J^1 Y^\star$ as before as an $(n+1)$-form on $Y$, and $u_{1},...,u_{n+1} \in T_{z} Z$. The {\bfi canonical multisymplectic  $(n+2)$-form} $\Omega$ on $J^1Y^\star$ is now defined as
\[
\Omega = -\mathbf{d}\Theta.
\]

In coordinates, one has the following expression for $\Theta$:
\[
\Theta=p_{A}^{\mu}dy^{A} \wedge d^{n}x_{\mu}+p\,d^{n+1}x,
\]
while $\Omega$ is locally given by 
\begin{equation} \label{msform}
\Omega=dy^{A} \wedge dp_{A}^{\mu}\wedge d^{n}x_{\mu}-dp \wedge d^{n+1}x.
\end{equation}
It is easy to show (see Ref.~\onlinecite{CaIbLe1999}) that $\Omega$ is non-degenerate in the sense that, for all $v \in TZ$, $\mathbf{i}_v \Omega=0$ implies that $v = 0$.  Moreover, $\Omega$ is trivially closed: $\mathbf{d} \Omega = 0$.  Forms which are both closed and non-degenerate are referred to as {\bfi multisymplectic} forms.  If the degree of $\Omega$ is two, $\Omega$ is just a symplectic form, while if $\Omega$ is of maximal degree, $\Omega$ is a volume form.  More examples of multisymplectic manifolds can be found in Ref.~\onlinecite{CaIbLe1999}.

A form which is non-degenerate but not necessarily closed is referred to as an {\bfi almost multisymplectic} form.  On the other hand, in the context of field theory, we will often refer to a closed but possibly degenerate form as a {\bfi pre-multisymplectic} form.

\subsection{Lagrangian Densities and the Covariant Legendre Transformation} 

A {\bfi Lagrangian density} is a smooth map $\mathcal{L}=L\eta:J^{1}Y \to \Lambda^{n+1}X$.  In local coordinates, we may write 
\[
\mathcal{L(\gamma)}=L(x^{\mu},y^{A},v^{A}_{\mu})d^{n+1}x,
\]
where $L$ is a function on $J^1 Y$ to which we also refer as the Lagrangian.

The corresponding {\bfi covariant Legendre transformation} for a given Lagrangian density $\mathcal{L}:J^{1}Y \to \Lambda^{n+1}X$ is a fiber preserving map $\mathbb{F}\mathcal{L}:J^{1}Y \to J^{1}Y^{\star}$ over $Y$, which is given by the first order vertical Taylor approximation to $\mathcal{L}$:
\[
\left<\mathbb{F}\mathcal{L}(\gamma),\gamma^{\prime}\right>=\mathcal{L}(\gamma) + \frac{d}{d \epsilon}\bigg|_{\epsilon=0} \mathcal{L}(\gamma+ \epsilon (\gamma-\gamma^{\prime})),
\]
where $\gamma, \gamma^{\prime} \in J^{1}Y$. In coordinates, $\mathbb{F}\mathcal{L}(x^\mu, y^A, v^A_\mu) = (x^\mu, y^A, p_A^\mu, p)$, where
\begin{equation}\label{CovLegTrans}
p_{A}^{\mu}=\frac{\partial{L}}{\partial{v^{A}_{\mu}}}, \qquad p=L-\frac{\partial{L}}{\partial{v^{A}_{\mu}}}v^{A}_{\mu}.
\end{equation}


\section{The Hamilton-Pontryagin Principle for Field Theories} \label{sec:hpprinciple}

In this section, we introduce the Hamilton-Pontryagin action principle for classical field theories.  Let us first define the Pontryagin bundle 
\[
M :=J^{1}Y \times_Y Z,
\]
as the fibered product over $Y$ of the jet bundle $J^1 Y$ and its dual $Z$ (recall that we identified $J^1 Y^\star$ with $Z$).  This bundle plays a similar role as the Pontryagin bundle $TQ \oplus T^{\ast}Q$ over a configuration manifold $Q$ for the case of classical mechanics.  Note that $M$ is a bundle over $Y$ and hence also over $X$.

\subsection{The Generalized Energy Density} 

Let $\mathcal{L}: J^{1}Y \to \Lambda^{n+1}X$ be a Lagrangian density which is possibly degenerate. We define the {\bfi generalized energy density} associated to $\mathcal{L}$ by the map  $\mathcal{E}:M \to \Lambda^{n+1}X$ defined as 
\begin{equation} \label{generalizedenergydensity}
\mathcal{E}(\gamma,z):=E(\gamma,z)\eta=\left< z, \gamma \right>-\mathcal{L}(\gamma)
\end{equation}
for $(\gamma,z) \in M$.  In local coordinates $(x^{\mu},y^{A},v^{A}_{\mu},p, p^{\mu}_{A})$ on $M$, the generalized energy density is represented by  
\begin{equation} \label{generalizedenergy}
\begin{split}
\mathcal{E}&=E(x^{\mu},y^{A},v^{A}_{\mu},p,p^{\mu}_{A})d^{n+1}x\\
&=\left(p+p_{A}^{\mu}v^{A}_{\mu}-L(x^{\mu},y^{A},v^{A}_{\mu})\right)d^{n+1}x,
\end{split}
\end{equation}
where $E=p+p_{A}^{\mu}v^{A}_{\mu}-L(x^{\mu},y^{A},v^{A}_{\mu})$ is called the {\bfi generalized energy} on $M$.

\subsection{Pre-Multisymplectic Forms on $M$} 

Recall the definition of the canonical forms $\Theta$ and $\Omega:=-\mathbf{d}\Theta$ on the dual jet bundle $Z \cong J^1 Y^\star$, and let $\pi_{Z M}: M \rightarrow Z$ be the projection onto the second factor.  The forms $\Theta$ and $\Omega$ can be pulled back along $\pi_{ZM}$ to yield corresponding forms on $M$:
\[
	\Theta_M := \pi_{ZM}^\ast \Theta \quad \text{and} \quad 
	\Omega_M := \pi_{ZM}^\ast \Omega.
\]
Note that $\Omega_M$ cannot be multisymplectic since it has a non-trivial kernel: for all {$v \in T J^1 Y$} we have that $\mathbf{i}_v \Omega_M = 0$.  As a result, we will refer to $\Omega_M$ as the {\bfi pre-multisymplectic $(n + 2)$-form} on $M$.  

Finally, for any Lagrangian density $\mathcal{L}$ with associated energy density $\mathcal{E}$, we introduce another pre-multisymplectic $(n + 2)$-form on $M$ by
\[
\Omega_{\mathcal{E}}=\Omega_{M}+\mathbf{d}\mathcal{E}.
\]

\subsection{The Hamilton-Pontryagin Principle} \label{sec:hp_principle}

Using the pre-multisymplectic forms $\Omega_M$ and $\Omega_{\mathcal{E}}$, we establish the {\bfi Hamilton-Pontryagin variational principle} for classical field theories.  This principle is similar to the expression \eqref{HPmech} for mechanical systems outlined in the introduction.

\begin{definition}\label{def:HPActionFunct}
Consider a Lagrangian density $\mathcal{L}$ with associated generalized energy density $\mathcal{E}$.  The Hamilton-Pontryagin action functional is defined as 
\begin{equation}\label{HPActionFunc}
	S(\psi) = \int_{X}  \psi^{\ast}(\Theta_{M} -\mathcal{E}),
\end{equation}
where $\psi$ is a section of $\pi_{XM}: M \to X$.  In other words, $\psi$ can be written as $\psi=(\gamma, z)$ where $\gamma$ is a section of $\pi_{X,J^{1}Y}:J^{1}Y \to X$ and $z$ is a section of $\pi_{XZ}: Z \to X$.
\end{definition}

We have defined the Hamilton-Pontryagin principle in terms of the  {$(n+1)$-} form $\Theta_M$.  By using the definition \eqref{canonnform} and \eqref{generalizedenergydensity}, we can rewrite the Hamilton-Pontryagin action functional as 
\[
	S(\psi) = \int_X \left( \left<z, j^1 \phi \right> 
		- \left<z, \gamma \right> + \mathcal{L}(\gamma) 
			\right),
\]
where $\phi$ is a section of $\pi_{XY}$ as before. An expression for $S$ in local coordinates will be given below in Section~\ref{sec:coords}.

A {\bfi vertical variation} of a section $\psi$ of $\pi_{XM}: M \to X$ is a one-parameter family $\psi_\lambda:X \to M$ of sections, such that $\psi_0=\psi$, and $\psi_\lambda$ agrees with $\psi$ outside a compact set $U$ of $X$.   Such a variation can be constructed by considering a one-parameter family of diffeomorphisms $\nu_\lambda : M \rightarrow M$ such that $\nu_0$ is the identity, $\nu_\lambda$ is the identity outside of $\pi_{XM}^{-1}(U)$ and $\nu_\lambda$ preserves the fibration $\pi_{XM} : M \rightarrow X$, or in other words $\pi_{XM} \circ \nu_\lambda = \pi_{XM}$.  Then, the composition $\psi_\lambda=\nu_\lambda \circ \psi$ is a vertical variation of $\psi$.

At the infinitesimal level, an infinitesimal variation of $\psi$ is a $\pi_{XM}$-vertical vector field $\mathcal{V}_{M}: \mathrm{Im}\, \psi \to VM$, defined on the image of $\psi$ in $M$.  Here $VM$ is the vertical subbundle of $TM$, defined by
\[
VM=\left\{ w \in TM \mid T\pi_{XM} (w) =0\right\}.
\]
We denote by $\mathfrak{X}^{\mathrm{V}}(M)$ the module of the $\pi_{XM}$-vertical vector fields on $M$.

We say that $\psi$ is a {\bfi critical point} of the action \eqref{HPActionFunc} if  $\delta S(\psi) = 0$ for all variations of $\psi$, where 
\[
\begin{split}
\delta S(\psi)&=\frac{d}{d\lambda}\bigg|_{\lambda=0}S(\psi_{\lambda})\\
&=\frac{d}{d\lambda}\bigg|_{\lambda=0} \int_{U} \psi_{\lambda}^{\ast} ( \Theta_{M} -\mathcal{E}),
\end{split}
\]
where $U$ is a compact subset of $X$ that depends on the variation $\psi_\lambda$.

\begin{proposition} \label{prop:implEL}
A section $\psi$ of $M$ is a critical point of the Hamilton-Pontryagin action functional if and only if $\psi$ satisfies the 
{\bfi implicit Euler-Lagrange equations}
\begin{equation}\label{ELFieldEqn-M}
\begin{split}
\psi^{\ast}(\mathcal{V}_{M} \!\text{\Large$\lrcorner$} \;\Omega_{\mathcal{E}})=0,  \quad \mbox{for any \;$\mathcal{V}_{M} \in \mathfrak{X}^{\mathrm{V}}(M)$}.
\end{split}
\end{equation}
\end{proposition}

\begin{proof}
Recall that a variation of $\psi$ is given by $\psi_{\lambda}=\nu_{\lambda} \circ \psi$ where $\nu_{\lambda}: M \to M$ is the flow of a vertical vector field $\mathcal{V}_{M}$. Then, it follows that
\begin{equation*}
\begin{split}
\frac{d}{d\lambda}\bigg|_{\lambda=0}S(\psi_{\lambda})&=\frac{d}{d\lambda}\bigg|_{\lambda=0}\int_{U} \psi^{\ast}_{\lambda} (\Theta_{M} - \mathcal{E})\\
&=\int_{U} \psi^{\ast} \pounds_{{\mathcal{V}}_{M}} (\Theta_{M} -\mathcal{E})\\
&=-\int_{U} \psi^{\ast} (\mathcal{V}_{M} \! \text{\Large$\lrcorner$} \; (\Omega_{M} + \mathbf{d} \mathcal{E})) 
+ \int_U  \mathbf{d}( \psi^\ast (\mathcal{V}_M\! \text{\Large$\lrcorner$} \;
(\Theta_M - \mathcal{E})))
\\
&=-\int_{U} \psi^{\ast} (\mathcal{V}_{M} \! \text{\Large$\lrcorner$} \;\Omega_{\mathcal{E}})
+ \int_{\partial U}  \psi^\ast (\mathcal{V}_M\! \text{\Large$\lrcorner$} \;
(\Theta_M - \mathcal{E}))
\\
&=-\int_{U} \psi^{\ast} (\mathcal{V}_{M} \! \text{\Large$\lrcorner$} \;\Omega_{\mathcal{E}})
\end{split}
\end{equation*}
for all $\mathcal{V}_{M}$,
where we utilized Cartan's magic formula and Stokes' theorem and the fact that the $\pi_{XM}$-vertical vector field $\mathcal{V}_{M}: M \to VM$  is compactly supported in $U \subset X$.  A standard argument then shows that $\psi$ is a critical point of $S$ if and only if \eqref{ELFieldEqn-M} holds.
\end{proof}

\subsection{Coordinate Expressions} \label{sec:coords}

Employing local coordinates $(x^{\mu},y^{A},v^{A}_{\mu},p,p_{A}^{\mu})$ on $M$, the action functional is denoted by
\begin{align*}
S(\gamma)&=\int_{U}  \psi^{\ast}(\Theta_{M} -\mathcal{E})  \nonumber   \\
&=\int_{U} \left(p+p_{A}^{\mu} \frac{\partial y^{A}}{\partial x^{\mu}} \right)d^{n+1}x- \left\{\left(p+p_{A}^{\mu}v^{A}_{\mu})d^{n+1}x-L(x^{\mu},y^{A},v^{A}_{\mu}\right)d^{n+1}x \right\} \nonumber  \\
&=\int_{U} \left\{ p_{A}^{\mu} \left( \frac{\partial y^{A}}{\partial x^{\mu}}- v^{A}_{\mu} \right)+ L(x^{\mu},y^{A},v^{A}_{\mu})\right\} d^{n+1}x.
\end{align*}
A small coordinate calculation, similar to the one in the introduction, then shows that the implicit Euler-Lagrange equations \eqref{ELFieldEqn-M} are locally given by 
\begin{equation}\label{local-ELFieldEqn}
\frac{\partial y^{A}}{\partial x^{\mu}}=v^{A}_{\mu}, \quad \frac{\partial p_{A}^{\mu}}{\partial x^{\mu}}= \frac{\partial L}{\partial y^{A}}, \quad p_{A}^{\mu}= \frac{\partial L}{\partial v^{A}_{\mu}}.
\end{equation}

\subsection{Generalized Energy Constraint}

The coordinates on the Pontryagin bundle are $(x^\mu, y^A, v^A_\mu, p, p_A^\mu)$, and the Hamilton-Pontryagin principle gives equations for the fields $y^A$, $v^A_\mu$ and $p_A^\mu$, but leaves $p$ undetermined.  In order to fix $p$, we impose the {\it generalized energy constraint} 
\[
E \equiv p+p_{A}^{\mu}v^{A}_{\mu}-L(x^{\mu},y^{A},v^{A}_{\mu})=0.
\]
Together with the third equation in \eqref{local-ELFieldEqn}, this naturally recovers the covariant Legendre transformation \eqref{CovLegTrans}.  We summarize the results up to this point in the following theorem.

\begin{theorem} \label{thm:field_equations}
The following statements for a section $\psi : X \rightarrow M$ of $\pi_{XM}$ satisfying $E(\psi) = 0$ are equivalent:
\begin{itemize}
\item[(1)]  $\psi$ is a critical point of the Hamilton-Pontryagin action functional \eqref{HPActionFunc}.
\item[(2)]  $\psi^{\ast} (\mathcal{V}_{M} \! \text{\Large$\lrcorner$} \;\Omega_{\mathcal{E}})=0$ for all $\pi_{XM}$-vertical vector fields $\mathcal{V}_{M}$ on $M$.
\item[(3)] $\psi$ satisfies the implicit Euler-Lagrange equations \eqref{local-ELFieldEqn} together with the covariant Legendre transformation \eqref{CovLegTrans}.
\end{itemize}
\end{theorem}

\subsection{Implicit Field Equations Using Multivector Fields}

We now rewrite the field equations \eqref{local-ELFieldEqn} in a slightly different way, which will make the link with multi-Dirac structures in Section~\ref{sec:multidirac} clearer.  Up to this point, we have viewed the field equations as coming from the Hamilton-Pontryagin variational principle.  We now show how these equations can be understood directly in terms of the multisymplectic form.  A quick introduction to the geometry of multivector fields can be found in appendix~\ref{sec:multivector}.

To this end, we will use multivector fields on $M$ (see Refs.~\onlinecite{Ma1997, EcMuRo2002, FoPaRo2005}), which can be seen as the infinitesimal counterparts of sections $\psi : X \to M$ of the Pontryagin bundle, in the sense that a multivector field gives rise to a set of first-order PDEs, whose integral sections are sections of the Pontryagin bundle.  For the purpose of this paper, we will restrict our attention to multivector fields that are locally of the form 
\begin{equation} \label{multvectf} 
	\mathcal{X} = \bigwedge_{\mu = 1}^{n + 1} \mathcal{X}_{\mu} = \bigwedge_{\mu = 1}^{n + 1}
		\left( \frac{\partial}{\partial x^\mu}
			+ C^A_\mu \frac{\partial}{\partial y^A}
			+ C_{A\mu}^\nu \frac{\partial}{\partial p_A^\nu}
			+ C_\mu \frac{\partial}{\partial p} \right),
\end{equation}
where $C^A_\mu$, $C_{A\mu}^\nu$ and $C_\mu$ are local component functions of $\mathcal{X}$ on $M$.  Note in particular that the component  of the multivector field in the $v^A_\mu$-direction  is zero.  We will refer to multivector fields $\mathcal{X}$ which are locally given by \eqref{multvectf} as {\bfi partial multivector fields}.  

The characterization \eqref{multvectf} can be recast into an intrinsic form by noting that a partial multivector field $\mathcal{X}$ is a map from $M$ into $\bigwedge^{n+1} (TZ)$ with the following properties: 
\begin{itemize}
	\item $\mathcal{X}$ is \emph{decomposable}:  in other words, there exist $n + 1$ vector fields $\mathcal{X}_\mu$ on $M$ so that $\mathcal{X}$ can be written as the wedge product $\mathcal{X} = \wedge_{\mu = 1}^{n + 1} \mathcal{X}_\mu$.
	
	\item $\mathcal{X}$ is \emph{normalized}, in the sense that $\mathbf{i}_{\mathcal{X}} (\pi^\ast_{XM} \eta) = 1$ with $\eta$ the volume form on $X$.  As a result, the coefficient of $\partial/\partial x^\mu$ in \eqref{multvectf} is unity.
\end{itemize}

Since $TZ$ is a subbundle of $TM = T(J^1 Y) \times_{TY} TZ$, we can identify a partial multivector field $\mathcal{X}$ with a regular multivector field on $M$ which is \emph{$J^1 Y$-vertical}, i.e. which has no component along $\partial/\partial v^A_\mu$.  More information about multivector fields and their properties can be found in Refs.~\onlinecite{EcMuRo1998, EcMuRo2002}.

A section $\psi: X \to M$, given in local coordinates by 
\[
 \psi(x) = (x^\mu, y^A(x), y^A_\mu(x), p(x), p_A^\mu(x)), 
\]
is an {\bfi integral section} of $\mathcal{X}$ if it satisfies the following system of PDEs:
\begin{equation} \label{intsection}
	\frac{\partial y^A}{\partial x^\mu} = C^A_\mu, \quad
	\frac{\partial p_A^\nu}{\partial x^\mu} = C_{A\mu}^\nu, \quad 
	\frac{\partial p}{\partial x^\mu} = C_\mu,
\end{equation}
which express the condition that the image of $\psi$ is tangent to $\mathcal{X}$.  We now claim that the implicit field equations \eqref{local-ELFieldEqn} can be expressed using multivector fields and the generalized energy \eqref{generalizedenergy} as follows:

\begin{theorem} \label{thm:implicitsympl}
Let $E$ be the generalized energy \eqref{generalizedenergy} and consider a multivector field $\mathcal{X}$ on $M$ of the form \eqref{multvectf} such that 
\begin{equation} \label{imp_EL}
\mathbf{i}_{\mathcal{X}}\Omega_{M}=(-1)^{n+2}\mathbf{d}E,
\end{equation}
where $\dim X = n + 1$, and let $\psi: X \to M$ be an integral section of $\mathcal{X}$.  In local coordinates, where $\psi(x) = (x^\mu, y^A(x), y^A_\mu(x), p(x), p_A^\mu(x))$, we then have that $\psi$ satisfies the following set of local equations:
\begin{equation} \label{eqm-LDS}
	\frac{\partial p_A^\mu}{\partial x^\mu} =  \frac{\partial L}{\partial y^A}, \quad
	\frac{\partial y^A}{\partial x^\mu} = v^A_\mu, \quad
	p_A^\mu = \frac{\partial L}{\partial v^A_\mu},
\end{equation}
together with
\begin{equation} \label{encons-LDS}
	\frac{\partial}{\partial x^\mu} \left( p+p_A^\nu v^A_\nu - L \right) = 0.
\end{equation}
\end{theorem}
\begin{proof}
To show this, we compute the differential of the generalized energy \eqref{generalizedenergy}:
\[
\begin{split}
\mathbf{d}E&=\frac{\partial E}{\partial x^{\mu}}dx^{\mu}+\frac{\partial E}{\partial y^{A}}dy^{A}+\frac{\partial E}{\partial v^{A}_{\mu}}dv^{A}_{\mu}+\frac{\partial E}{\partial p_{A}^{\mu}}dp_{A}^{\mu}+\frac{\partial E}{\partial p}dp\\
&=\left(-\frac{\partial L}{\partial x^{\mu}}\right)dx^{\mu}+\left(-\frac{\partial L}{\partial y^{A}} \right)dy^{A}+\left(p_{A}^{\mu}-\frac{\partial L}{\partial v^{A}_{\mu}} \right)dv^{A}_{\mu}
+v^{A}_{\mu}dp_{A}^{\mu}+dp,
\end{split}
\]
and the contraction of the multivector field $\mathcal{X}$ given in \eqref{multvectf} with the multisymplectic form $\Omega_M$:
\begin{equation} \label{contract}
	\mathbf{i}_\mathcal{X} \Omega_{M} = 
		(-1)^{n+2} \left[
			(C^A_\mu C_{A\lambda}^\lambda - 
			 C^A_\lambda C^\lambda_{A\mu} - C_\mu) dx^\mu 
			 + C^A_\mu dp_A^\mu + C_{A\mu}^\mu dy^A + dp
		\right].
\end{equation}
As a result, we have that \eqref{imp_EL} holds if and only if 
\begin{equation} \label{comp1}
	C^A_\mu C_{A\lambda}^\lambda - C^A_\lambda C^\lambda_{A\mu}
	- C_\mu =  -\frac{\partial L}{\partial x^\mu}
\end{equation}
as well as 
\begin{equation} \label{comp2}
	C^A_\mu = v^A_\mu, 
		\quad 
	C_{A\mu}^\mu = \frac{\partial L}{\partial y^A}, 
		\quad
	p_A^\mu = \frac{\partial L}{\partial v^A_\mu}.
\end{equation}
Now, let $\psi$ be an integral section of $\mathcal{X}$, so that locally \eqref{intsection} holds.  From \eqref{comp2}, it follows that the component functions of $\psi$  satisfy
\[
	\frac{\partial y^A}{\partial x^\mu} = v^A_\mu, \quad \text{and} \quad \frac{\partial p_A^\mu}{\partial x^\mu} = \frac{\partial L}{\partial y^A},
\]
and together with the third relation in \eqref{comp2} these equations are just the Euler-Lagrange equations \eqref{eqm-LDS} in implicit form.  The expression \eqref{comp1} can now be rewritten by substituting into it the expression for $C^\lambda_{A\lambda}$ from \eqref{comp2} and rewriting the second term using the relation $C^A_\mu = v^A_\mu$ to obtain
\[
	C^A_\mu \frac{\partial L}{\partial y^A} - v^A_\lambda C^\lambda_{A\mu}
	- C_\mu =  -\frac{\partial L}{\partial x^\mu}.
\]
We now substitute the expressions from \eqref{intsection} determining the component functions of $\psi$.  The result is 
\[
	\frac{\partial y^A}{\partial x^\mu} \frac{\partial L}{\partial y^A} - v^A_\lambda \frac{\partial p_A^\lambda}{\partial x^\mu} - \frac{\partial p}{\partial x^\mu}
	=  -\frac{\partial L}{\partial x^\mu}.
\]
Using the Leibniz rule for the second term and regrouping the resulting terms, we then get 
\[
	\frac{\partial L}{\partial x^\mu} + \frac{\partial y^A}{\partial x^\mu} \frac{\partial L}{\partial y^A} + 
		\frac{\partial v^A_\lambda}{\partial x^\mu} \frac{\partial L}{\partial v^A_\lambda}
	= \frac{\partial }{\partial x^\mu} \left( p + p_A^\mu v^A_\mu \right),
\]	
which is precisely \eqref{encons-LDS}.
\end{proof}


\section{Multi-Dirac Structures and Field Theories}  \label{sec:multidirac}

In this section, we show how the implicit field equations $\mathbf{i}_{\mathcal{X}}\Omega_{M}=(-1)^{n+2}\mathbf{d}E$ of  theorem~\ref{thm:implicitsympl} can be described in terms of \emph{multi-Dirac structures}, which are a graded analogue of the concept of Dirac structures introduced by Courant in Ref.~\onlinecite{Co1990}. In order to keep the technical side of the paper to a minimum, we focus on the applications of multi-Dirac structures to field theories, and we leave proofs and a more thorough discussion to Ref.~\onlinecite{VaYoLe2011}.  We also point out that there are alternative approaches to the construction of a ``field-theoretic analogue'' of Dirac structures; see for instance Refs.~\onlinecite{BaHoRo2010, Za2012}.

\subsection{Pairings between Multivectors and Forms}

While multi-Dirac structures can in principle be defined on an arbitrary manifold $M$, we will focus for now on multi-Dirac structures on the Pontryagin bundle $M \to X$.  Recall that the dimension of the base space $X$ is written as $n + 1$.

Consider the spaces $\bigwedge^l(TM)$ of $l$-multivector fields on $M$ and $\bigwedge^k(T^\ast M)$ of $k$-forms on $M$. For $k \ge l$, there is a natural pairing between elements of $\Sigma \in\bigwedge^k(T^\ast M)$ and $\mathcal{X} \in \bigwedge^l(TM)$ given by
\begin{equation}\label{pairing-mvf}
\left< \Sigma, \mathcal{X} \right> :=\mathbf{i}_{\mathcal{X}} \Sigma \in \mbox{$\bigwedge$}^{k-l}(T^\ast M).
\end{equation}

We now introduce the {\bfi graded Pontryagin bundle of degree $r$} over $M$ as follows:
\[
P_{r} = \mbox{$\bigwedge^r(TM)$} \oplus \mbox{$\bigwedge^{n+2-r}(T^\ast M)$},
\]
where $r=1,...,n+1$.

Using the pairing in equation \eqref{pairing-mvf}, let us define the following  antisymmetric and symmetric pairings between the elements of $P_{r}$ and $P_{s}$ as follows.  For $(\mathcal{X}, \Sigma) \in P_{r}$ and $(\bar{\mathcal{X}}, \bar{\Sigma}) \in P_{s}$, where $r,s=1, \ldots ,n+1$, we put
\[
\left<\left< (\mathcal{X},  \Sigma), (\bar{\mathcal{X}}, \bar{\Sigma}) \right>\right>_{-} :=\frac{1}{2}\left( \mathbf{i}_{\bar{\mathcal{X}}}\Sigma-(-1)^{rs}\mathbf{i}_{\mathcal{X}}\bar{\Sigma} \right)
\]
and
\[
\left<\left< (\mathcal{X},  \Sigma), (\bar{\mathcal{X}}, \bar{\Sigma}) \right>\right>_{+} :=\frac{1}{2}\left( \mathbf{i}_{\bar{\mathcal{X}}}\Sigma+(-1)^{rs}\mathbf{i}_{\mathcal{X}}\bar{\Sigma} \right),
\]
each of which takes values in $\bigwedge^{n+2-r-s}(T^\ast M)$.  Hence, both of these pairing are identically zero whenever $n+2<r+s$.

Let $V_{s}$ be a subbundle of $P_{s}$. The {\bfi $r$-orthogonal complementary subbundle} of $V_{s}$ associated to the antisymmetric paring $\left<\left<,\right>\right>_{-}$  is the subbundle $(V_{s})^{\perp,r}$ of $P_{r}$ defined by
\[
\begin{split}
(V_{s})^{\perp,r}&= \{  (\mathcal{X},  \Sigma) \in P_{r} \mid 
\left<\left< (\mathcal{X},  \Sigma), (\bar{\mathcal{X}}, \bar{\Sigma}) \right>\right>_{-}=0 \quad \textrm{for all} \quad (\bar{\mathcal{X}}, \bar{\Sigma}) \in V_{s}
\}.
\end{split}
\]
Note that $(V_{s})^{\perp,r}$ is a subbundle of $P_r$, and that  $(V_{s})^{\perp,r}=P_r$ whenever $n+2 < r+s$.

\subsection{Almost Multi-Dirac Structures on Manifolds}  

The definition of an \emph{almost multi-Dirac structure} on $M$ mimics the standard definition of Ref.~\onlinecite{Co1990} of Dirac structures.  

\begin{definition}\label{DefMultiDirac} 
An {\bfi almost multi-Dirac structure of degree $n + 1$} on $M$ is a direct sum $D_1 \oplus \cdots \oplus D_{n+1}$ of subbundles 
\[
D_{r} \subset  P_{r} \quad \text{(for $r = 1, \ldots, n+1$)},
\]
which is {\bfi $(n+1)$-Lagrangian}; namely, the sequence of the bundles $D_{r}$ are {\bfi maximally $(n+1)$-isotropic:}
\begin{equation} \label{isotropy}
D_{r}=(D_{s})^{\perp, r}
\end{equation}
for $r, s=1,...,n+1$ and where $r+s \le n+2$. 
\end{definition}

When no confusion can arise, we will refer to the direct sum $D_1 \oplus \cdots \oplus D_{n+1}$ simply as $D$.   In the companion paper Ref.~\onlinecite{VaYoLe2011}, we define a graded version of the Courant bracket, and we refer to \emph{integrable multi-Dirac structures} as almost multi-Dirac structures which are closed under the graded Courant bracket.

For the case of classical field theories, it will turn out that only the orthogonal complements of the multi-Dirac structure of the ``lowest'' and ``highest'' order (namely, $r=1$ and $r=n+1$ respectively) play an essential role to formulate the field equations.

\begin{proposition}
Let $D = D_1 \oplus \cdots \oplus D_{n+1}$ be a multi-Dirac structure of degree $n + 1$.  For any $(\mathcal{X},\Sigma) \in D_{r}$ and $(\bar{\mathcal{X}},\bar{\Sigma}) \in D_{s}$, the following relation holds:
\begin{equation}\label{LagProp}
 \mathbf{i}_{\bar{\mathcal{X}}}\Sigma-(-1)^{rs}\mathbf{i}_{\mathcal{X}}\bar{\Sigma} =0.
\end{equation}
\end{proposition}
\begin{proof}
This is clear from the $r$-Lagrangian (maximally $r$-isotropic) property of $D_{r}$. 
\end{proof}

The above $r$-Lagrangian property of the multi-Dirac structure $D_{r}$ in equation \eqref{LagProp} may be understood as the field-theoretic analogue of the {\bfi virtual work principle} in mechanics and is related to {\bfi Tellegen's theorem} in electric circuits. 

\paragraph*{Example.}
Let $M$ be a manifold with a multi-Dirac structure $D_1$ of degree $1$.  The isotropy property then becomes
\[
	D_1 = (D_1)^{\perp, 1},  
\]
where $P_1 = TM \oplus T^\ast M$ and the pairing is given by 
\[
	\left<\left< (\mathcal{X},  \Sigma), (\bar{\mathcal{X}}, \bar{\Sigma}) \right>\right>_{-} =
	\frac{1}{2}\left( \mathbf{i}_{\bar{\mathcal{X}}}\Sigma + \mathbf{i}_{\mathcal{X}}\bar{\Sigma} \right),
\]
for all $(\mathcal{X}, \Sigma), (\bar{\mathcal{X}}, \bar{\Sigma}) \in P_1$.  This is nothing but the definition of an (almost) Dirac structure developed in Ref.~\onlinecite{Co1990}.  Our concept of multi-Dirac structures also includes, and is in fact equivalent to, the so-called \emph{higher-order Dirac structures} of Ref.~\onlinecite{Za2012}.

\subsection{Multi-Dirac Structures Induced by Differential Forms}

Consider an $(n+2)$-form $\Omega_{M}$ on a manifold $M$.  We will show that the graph of $\Omega_M$ (in the sense defined below) defines a multi-Dirac structure of degree $n + 1$.  This example of a multi-Dirac structure will be used in our subsequent treatment of classical field theories, where $\Omega_M$ will be the canonical multisymplectic form, but for now $\Omega_M$ can be an arbitrary form.  

\begin{proposition} \label{prop: canondirac}
Let $\Omega_{M}$ be an arbitrary $(n + 2)$-form on $M$ and define the following subbundles $D_1, \ldots, D_{n + 1}$, where $D_r \subset P_r$ is given by
\begin{equation}
\label{canonms}
D_{r} =\left\{ (\mathcal{X}, \Sigma)  \in P_{r}  \mid  
 \;\mathbf{i}_{\mathcal{X}} \Omega_{M}=\Sigma  \right\}
\end{equation}
for $r = 1, \ldots, n + 1$.  Then $D = D_1 \oplus \cdots \oplus D_{n + 1}$ is a  multi-Dirac structure of degree $n + 1$ on $M$.
\end{proposition}

\begin{proof}
To prove that $D = D_1 \oplus \cdots \oplus D_{n + 1}$ is a multi-Dirac structure of degree $n + 1$, we need to check the isotropy property \eqref{isotropy}, namely, $D_{r}=D_{s}^{\perp,r}$ for all $r, s = 1, \ldots, n + 1$, with $r + s \le n + 1$.

Let us first show that $D_{r} \subset D_{s}^{\perp,r}$. Let $(\mathcal{X},\Sigma) \in D_{r}$
and $(\bar{\mathcal{X}},\bar{\Sigma}) \in D_{s}$. Then, it follows that
\[
\begin{split}
&\left<\left< (\mathcal{X},\Sigma),(\bar{\mathcal{X}},\bar{\Sigma})\right>\right>_{-}\\
&\hspace{1cm}=\frac{1}{2}\left\{\mathbf{i}_{\bar{\mathcal{X}}} \Sigma +(-1)^{rs+1}\mathbf{i}_{\mathcal{X}} \bar{\Sigma}\right\}\\
&\hspace{1cm}=\frac{1}{2}\left\{\mathbf{i}_{\bar{\mathcal{X}}}\mathbf{i}_{\mathcal{X}} \Omega_{M} +(-1)^{rs+1} \mathbf{i}_{\mathcal{X}}\mathbf{i}_{\bar{\mathcal{X}}}\Omega_{M}\right\}
\\
&\hspace{1cm}=0,
\end{split}
\]
since $\mathbf{i}_{\bar{\mathcal{X}}}\mathbf{i}_{\mathcal{X}}\Omega_{M}=(-1)^{rs} \mathbf{i}_{\mathcal{X}}\mathbf{i}_{\bar{\mathcal{X}}}\Omega_{M}$.
Thus, $D_{r} \subset D_{s}^{\perp,r}$.
\medskip

Next, let us show that $D_{s}^{\perp,r} \subset D_{r}$. Let $(\bar{\mathcal{X}}, \bar{\Sigma}) \in D_{s}^{\perp,r}$. By definition of $D_{s}^{\perp,r}$,
\[
\mathbf{i}_{\bar{\mathcal{X}}} \Sigma + (-1)^{rs+1}\mathbf{i}_{\mathcal{X}} \bar{\Sigma}=0
\] 
for all $(\mathcal{X},\Sigma)  \in D_{r}$, i.e. 
$\mathcal{X} \in \bigwedge^{r}(TZ)$ such that $\mathbf{i}_{\mathcal{X}} \Omega_{M}=\Sigma$.  Then, it follows
\[
\begin{split}
\mathbf{i}_{\bar{\mathcal{X}}} \Sigma + (-1)^{rs+1}\mathbf{i}_{\mathcal{X}} \bar{\Sigma} &=\mathbf{i}_{\bar{\mathcal{X}}} \mathbf{i}_{\mathcal{X}} \Omega_{M} + (-1)^{rs+1}\mathbf{i}_{\mathcal{X}} \bar{\Sigma} \\
&=\mathbf{i}_{\mathcal{X}} \left\{  (-1)^{rs} \mathbf{i}_{\bar{\mathcal{X}}} \Omega_{M}+ (-1)^{rs+1} \bar{\Sigma}\right\}\\
&=0
\end{split}
\]
for all $\mathcal{X} \in \bigwedge^{r}(TZ)$ where $r+s \le n+2$. Therefore, one has
\[
\mathbf{i}_{{\bar{\mathcal{X}}}}\Omega_{M}=\bar{\Sigma}.
\]

Thus, $D_{s}^{\perp, r} \subset D_{r}$ and it follows that $D_{r}=D_{s}^{\perp, r}$, so that $D = D_1 \oplus \cdots \oplus D_{n + 1}$ is a multi-Dirac structure of degree $n + 1$ on $M$.
\end{proof}

\paragraph*{Local Characterization of Multi-Dirac Structures.}

In Ref.~\onlinecite{Za2012}, Zambon  has given a complete local characterization of multi-Dirac structures.  The following theorem is a combination of corollary~3.9 and lemma~4.3 in that paper.

\begin{theorem} \label{thm:localchar}
	Let $T$ be an arbitrary vector space and consider a form $\omega \in \bigwedge^{n+1} T^\ast$ and a subspace $S \subset T$ such that $\dim S \le \dim T - n$ or $S = T$.  Then the sets 
	\[
		D_r = \{ (X, \Sigma)  \mid\; {X} \in S \wedge (\wedge^{r-1} T), \quad \Sigma - \mathbf{i}_X \omega \in \wedge^{n+1-r} S^\circ \},
	\]
	$r = 1, \ldots, n$, 
	form a multi-Dirac structure.  Conversely, each multi-Dirac structure on $T$ is of this form, for a suitable choice of $S$ and $\omega$.
\end{theorem}

\subsection{Applications to Implicit Field Theories}

We now consider Lagrangian field theories in the context of multi-Dirac structures.  Consider a fiber bundle $\pi_{XY}: Y \to X$ and let $M = J^1 Y \times_Y Z$ be the Pontryagin bundle.  Let $\mathcal{L}$ be a Lagrangian density on $J^1 Y$ and denote the corresponding generalized energy density by $\mathcal{E}$.  Let $\Omega$ be the canonical multisymplectic structure on $Z$ and let $\Omega_M$ be the pre-multisymplectic $(n+2)$-form defined on the Pontryagin bundle by $\Omega_{M}=\pi_{ZM}^{\ast}\Omega$.
  
Using $\Omega_M$ we can define a multi-Dirac structure $D = D_1 \oplus \cdots \oplus D_{n + 1}$ of degree $n + 1$ on $M$ by the construction of proposition~\ref{prop: canondirac}.  We refer to $D$ as the {\bfi canonical multi-Dirac structure} on $M$.  Explicitly, $D = D_1 \oplus \cdots \oplus D_{n + 1}$ is given by 
\[
	D_r = \{ (\mathcal{X}, \mathbf{i}_{\mathcal{X}} \Omega_M) \mid  \mathcal{X} \in \mbox{$\bigwedge$}^r(TM) \},
\]
for $r = 1, \ldots, n + 1$.   As $\mathbf{d} \Omega_M = 0$, this multi-Dirac structure turns out be \emph{integrable}; see Ref.~\onlinecite{VaYoLe2011}.

From theorem~\ref{thm:implicitsympl}, we have that the implicit field equations can be written in terms of a partial multivector field $\mathcal{X}$ and the generalized energy $E$ as 
\[
	\mathbf{i}_{\mathcal{X}}\Omega_{M}=(-1)^{n+2}\mathbf{d}E.
\]
This can be expressed in terms of the multi-Dirac structure induced by $\Omega_M$ by the condition $(\mathcal{X}, (-1)^{n+2}\mathbf{d}E) \in D_{n+1}$.  Note that the other components $D_1, \ldots, D_n$ of the multi-Dirac structure are \emph{not needed to write down the field equations}.

Taking our cue from this observation, we now define a special class of field theories, whose field equations are specified in terms of a Lagrangian $\mathcal{L}$, a multivector field $\mathcal{X}$, and a multi-Dirac structure $D$ on the Pontryagin bundle $M$.

\begin{definition} \label{def:lagrangedirac}
Let $D = D_1 \oplus \cdots \oplus D_{n + 1}$ be a multi-Dirac structure of degree $n + 1$ on the Pontryagin bundle $M$ and consider a Lagrangian density $\mathcal{L}$ with associated generalized energy $E$ given by \eqref{generalizedenergy}.  A {\bfi Lagrange-Dirac system for field theories} is defined by a triple $(\mathcal{X},E, D_{n+1})$, where $\mathcal{X}$ is a normalized, decomposable partial vector field, so that 
\[
(\mathcal{X}, (-1)^{n+2}\mathbf{d}E) \in D_{n+1}.
\]
\end{definition}
Note that in the definition of a Lagrange-Dirac system, only the highest-order component $D_{n + 1}$ of the multi-Dirac structure appears.  

\begin{theorem} \label{thm:energy}
For any Lagrange-Dirac system $(\mathcal{X},E,D_{n+1})$, the condition of  {\bfi conservation of the generalized energy}
\begin{equation} \label{econs}
	\mathcal{X} \prodint \mathbf{d} E = 0
\end{equation}
is satisfied.
In other words,  the generalized energy $E$ is constant along the solutions of the Lagrange-Dirac system $(\mathcal{X},E,D_{n+1})$ (i.e. the integral manifolds of $\mathcal {X}$).
\end{theorem}

\begin{proof}
The proof relies on lemma~\ref{lemma:decomp} for decomposable vector fields.  Using the field equations, the inner product on the left-hand side of \eqref{econs} can be written as 
\[
	\mathcal{X} \prodint \mathbf{d} E = \mathcal{X} \prodint (\mathcal{X} \intprod \Omega_{M}) = 0,
\]
according to Corollary~\ref{cor:prodprod}.
\end{proof}

To see why the previous theorem implies energy conservation, decompose the multivector field $\mathcal{X}$ as in \eqref{multvectf}.  Using lemma~\ref{lemma:decomp}, the interior product can then be written as 
\[
	\mathcal{X} \prodint \mathbf{d} E = 
	\sum_{\mu = 1}^k (-1)^{\mu + 1} \left<\mathcal{X}_\mu, 
		\mathbf{d} E \right> 
				\hat{\mathcal{X}}_\mu
\]
where $\hat{\mathcal{X}}_\mu$ is the $n$-multivector field obtained by deleting $\mathcal{X}_\mu$ from $\mathcal{X}$, i.e., 
\[
		\hat{\mathcal{X}}_\mu = \bigwedge_{\stackrel{\lambda = 1}{\lambda \ne \mu}}^{n+1} \mathcal{X}_\lambda.
\]
Since the multivector fields $\hat{\mathcal{X}}_\mu$ are linearly independent, the energy conservation equation \eqref{econs} then implies that $\left<\mathcal{X}_\mu, \mathbf{d} E \right> = 0$ for $\mu = 1, \ldots, n + 1$.  In other words, the function $E$ is constant along the integral manifolds of $\mathcal{X}$.


\section{Examples} \label{sec:examples}

In this section, we shall present some examples of Lagrange-Dirac field theories in the context of multi-Dirac structures. Namely, we show that the two examples of {\it the scalar nonlinear wave equation} and {\it electromagnetism} can be described as standard Lagrange-Dirac field theories. Furthermore, as an example of the Lagrange-Dirac field theories with nonholonomic constraints, we consider {\it time-dependent mechanical systems with affine constraints}.  Lastly, we show that the {\it Hu-Washizu variational principle} of linear elastostatics can be understood {in the context of the} Hamilton-Pontryagin principle.

\subsection{Nonlinear Wave Equations}

Let us consider the scalar nonlinear wave equation in two dimensions as discussed in Ref.~\onlinecite{Br1997}.  The configuration bundle is $\pi_{XY}: Y \rightarrow X$, with $X = \mathbb{R}^2$, so that $n = \dim X - 1 = 1$, and ${Y = X \times \mathbb{R}}$.  Local coordinates are given by $(t, x)$ for $U \subset X$ and $(t, x,\phi)$ for $Y$, and hence $(t, x ,\phi,v_{t}, v_{x})$ for $J^{1}Y$ and $(t, x,\phi,p,p^{t}, p^{x})$ for $Z \cong J^{1}Y^{\star}$.  The nonlinear wave equation is given by 
\[
\frac{\partial^{2} \phi }{\partial t^{2}} - \frac{\partial^{2} \phi }{\partial x^{2}} -V^{\prime}(\phi)=0,
\]
where $\phi$ is a section of $\pi_{XY}$ and $V : Y \rightarrow \mathbb{R}$ is a nonlinear potential.  The Lagrangian for this equation is given by
\[
L(t, x ,\phi,v_{t}, v_{x})=\frac{1}{2}\left(v_{t}^{2} -v_{x}^{2}\right) +V(\phi) 
\]
and the Lagrangian density is of course given by $\mathcal{L} = L dt \wedge dx$.  The generalized energy is then $E=p+ p^{t} v_{t} +p^{x} v_{x} - L(t, x,\phi,v_{t}, v_{x})$.

\subsubsection{Hamilton-Pontryagin Principle}  

The Hamilton-Pontryagin action functional \eqref{HPActionFunc} for the scalar field is given by
\[
S(\psi) = 
 \int_{U} \left( \frac{1}{2}\left(v_{t}^{2} -v_{x}^{2}\right) +V(\phi) + p^{t}( \phi_{,t} - v_{t}) + p^{x} (\phi_{, x} - v_{x}) \right)dt \wedge dx.
\]
Taking variations of $\phi, v_t, v_x, p^t, p^x$ (with the condition that $\delta \phi$ vanishes on the boundary of $U$), we obtain the implicit field equations
\begin{equation}\label{ImpEulLagEq-NW}
\frac{\partial \phi}{\partial t}=v_{t}, \quad 
\frac{\partial \phi}{\partial x}=v_{x},  \quad p^t=v_{t},  \quad p^{x}=-v_{x},  \quad 
\frac{\partial p^t}{\partial t}
+
\frac{\partial p^x}{\partial x}
=V^{\prime}(\phi), 
\end{equation}
and imposing the generalized energy constraints $E=0$, one has
\begin{equation}\label{Gen-energy-const}
p=L(t, x,\phi,v_{t},v_{x})-p^{t}v_{t}-p^{x}v_{x}.
\end{equation}
The equations \eqref{ImpEulLagEq-NW} are precisely the covariant equations obtained in Ref.~\onlinecite{Br1997}.

\subsubsection{Multi-Dirac Structures}

The canonical pre-multisymplectic 3-form on $M$ is given by
\[
\Omega_{M}=-dp^{t}\wedge d\phi \wedge dx+dp^{x}\wedge d\phi \wedge dt - dp \wedge dt \wedge dx,
\]
and one can define a multi-Dirac structure $D= D_{1} \oplus D_{2}$ of degree $2$ associated with $\Omega_M$ on $M$ by the standard construction, where we set $n=1$.  We omit the definition of $D_1$, since only $D_2$ is important for the formulation of the dynamics:
\[
D_{2}=\left\{ (\mathcal{X},\Sigma) \in P_{2}= \mbox{$\bigwedge$}^2(TM) \oplus T^\ast M \mid \mathbf{i}_{\mathcal{X}}\Omega_{M}=\Sigma \right\}.
\]
Here, $\mathcal{X}$ is a partial multivector field locally given by 
\[
\mathcal{X} 
	= \bigwedge_{\mu = 0}^{1}
		\left( \frac{\partial}{\partial x^\mu}
			+ C_\mu \frac{\partial}{\partial \phi}
			+ C_{\mu}^\nu \frac{\partial}{\partial p^\nu}			
			+ c_\mu \frac{\partial}{\partial p} \right),
\]
where the variables $x^\mu$, $\mu = 0, 1$, are given by $x^0 = t$ and $x^1 = x$.  The field equations can now be expressed as 
\[
{
(\mathcal{X},-\mathbf{d}E) \in D_{2},
}
\]
or in other words $\mathbf{i}_{\mathcal{X}} \Omega_{M} =- \mathbf{d}E$.  In coordinates, these equations become
\[
\begin{split}
&v_{0}=C_{0}, \qquad v_{1}=C_{1},\qquad C^{0}_{0}+C^{1}_{1}=\frac{\partial L}{\partial \phi}, \qquad p^{0}=\frac{\partial L}{\partial v_{0}},\qquad p^{1}=\frac{\partial L}{\partial v_{1}}, \\
&c_{0}+C^{1}_{0}C_{1}-C_{0}C^{1}_{0}=\frac{\partial L}{\partial x}^{0}, \qquad c_{1}+C^{0}_{0}C_{1}-C_{0}C^{0}_{1}=\frac{\partial L}{\partial x^{1}},
\end{split}
\]
where
\[
C_{\mu}=\frac{\partial \phi}{\partial x^{\mu}}, \qquad C^{\nu}_{\mu}=\frac{\partial p^{\nu}}{\partial x^{\mu}}, \qquad c_{\mu}=\frac{\partial p}{\partial x^{\mu}}.
\]
Hence, by simple rearrangements, we obtain the implicit Euler-Lagrange equation \eqref{ImpEulLagEq-NW} for the nonlinear wave equation, together with
\[
\frac{\partial p}{\partial x^{\mu}}=\frac{\partial }{\partial x^{\mu}} (L(x^{0},x^{1},\phi,v_{0},v_{1})-p^{0}v_{0}-p^{1}v_{1}).
\]
By imposing the generalized energy constraint $E=0$, we again obtain \eqref{Gen-energy-const}.

\subsection{Electromagnetism}

The multisymplectic description of electromagnetism can be found, among others, in Ref.~\onlinecite{GoIsMa1997}.  Here, we highlight the role of the Hamilton-Pontryagin variational principle and the associated multi-Dirac structure.  The electromagnetic potential $A = A_\mu dx^\mu$ is a section of the bundle $Y = T^\ast X$ of one-forms on spacetime $X$.  For the sake of simplicity, we let $X$ be $\mathbb{R}^4$ with the Minkowski metric, but curved spacetimes can be treated equally well.  The bundle $Y$ has coordinates $(x^\mu, A_\mu)$ while $J^1 Y$ has coordinates $(x^\mu, A_\mu, A_{\mu, \nu})$.  The electromagnetic Lagrangian density is given by 
\[
	\mathcal{L}(A, j^1 A) = -\frac{1}{4} \mathbf{d} A \wedge \ast \mathbf{d} A = -\frac{1}{4} F_{\mu\nu} F^{\mu\nu},
\]
where the Hodge star operator $\ast$ is associated to the Minkowski metric and where $F_{\mu\nu}$ is the anti-symmetrization $F_{\mu\nu} = A_{\mu, \nu} - A_{\nu, \mu}$.

\subsubsection{Hamilton-Pontryagin Principle}

The Hamilton-Pontryagin action principle is given in coordinates by 
\[
	S = \int_U \left[
	p^{\mu, \nu} \left( \frac{\partial A_\mu}{\partial x^\nu} - A_{\mu, \nu} \right)
	- \frac{1}{4} F_{\mu\nu} F^{\mu\nu}  \right] d^4x,
\]
where $U$ is an open subset of $X$.  The Hamilton-Pontryagin equations are given by
\begin{equation} \label{HPmaxwell}
	p^{\mu,\nu} = F^{\mu\nu}, 
	\quad A_{\mu, \nu} = \frac{\partial A_\mu}{\partial x^\nu}, 
	\quad \frac{d p^{\mu, \nu}}{d x^\nu} = 0,
\end{equation}
and by eliminating $p^{\mu, \nu}$ lead to Maxwell's equations: $\partial_\nu F^{\mu \nu} = 0$.

\subsubsection{Multi-Dirac Structures} 

The canonical multisymplectic form for electromagnetism is given in coordinates by 
\[
	\Omega_M = d A_\mu \wedge d p^{\mu, \nu} \wedge d^n x_\nu 
		- dp \wedge d^{n + 1} x.
\]
Let $D = D_1\oplus \cdots\oplus D_{4}$ be the multi-Dirac structure induced by this form and consider in particular the component $D_{4}$ of the highest degree. In terms of this bundle, the implicit Euler-Lagrange equations can be described by 
\[
	(\mathcal{X},-\mathbf{d} E) \in D_{4}, 
\]
where $\mathcal{X}$ is a partial vector field of degree $4$ given locally by 
\[
	\mathcal{X} = \bigwedge_{\mu = 1}^{4} 
		\left(
			\frac{\partial}{\partial x^\mu}
			+ C_{\nu\mu} \frac{\partial}{\partial A_\nu}
			+ C_{\mu}^{\kappa\nu} \frac{\partial}{\partial 
				p^{\kappa,\nu}}			
			+ C_\mu \frac{\partial}{\partial p}
		\right) 
\]
and $E$ is the generalized energy density given by 
\[
	E(x^\mu, A_\mu, A_{\mu, \nu}, p^{\mu, \nu}, p) = p + p^{\mu, \nu} A_{\mu, \nu} 
		+ \frac{1}{4} F_{\mu\nu} F^{\mu\nu}.
\]
As before, a coordinate computation then shows that the coefficients of $\mathcal{X}$ are given by 
\[
	C_{\mu \nu} = A_{\mu, \nu},
	\quad
	C_\nu^{\mu\nu} = \frac{\partial L}{\partial A_\mu} = 0,
	\quad 
	p^{\mu,\nu} = \frac{\partial L}{\partial A_{\mu, \nu}} = 
	F^{\mu\nu}
\]
and 
\[
	C_\mu + C_\mu^{\kappa \lambda} C_{\kappa \lambda} 
		- C_{\nu\mu}C^{\nu\lambda}_\lambda =
	\frac{\partial L}{\partial x^\mu} = 0,
\]
where
\[
	C_{\nu\mu} = \frac{\partial A_{\nu}}{\partial x^\mu}, 
	\quad 
	C^{\kappa\lambda}_\mu = 
		\frac{\partial p^{\kappa, \lambda} }{\partial x^\mu}, 
	\quad \text{and} \quad
	C_\mu = \frac{\partial p}{\partial x^\mu}.
\]

A similar computation as in the proof of Theorem~\ref{thm:implicitsympl} shows that these equations are equivalent to Maxwell's equations \eqref{HPmaxwell} in implicit form, together with the energy constraint \eqref{encons-LDS}.

\subsection{Time-Dependent Mechanical Systems with Affine Constraints}

We now show how time-dependent mechanics with linear nonholonomic constraints can be interpreted as a special case of Lagrange-Dirac field theories.  Let $Q$ be the configuration space of a mechanical system and consider the bundle $\pi_{XY}: Y \to X$, where $X = \mathbb{R}$ (coordinatized by the time $t$) and $Y = \mathbb{R} \times Q$, with coordinates $(t, q^i)$.  The first jet bundle is given  by $J^{1}Y \cong \mathbb{R} \times TQ$, with  local coordinates $(t,q^{i},v^{i})$.  Let $\mathcal{L}:J^{1}Y \to \Lambda^{1}(X)$ be a Lagrangian density and put $\mathcal{L(\gamma)}=L(t,q^{i},v^{i})dt$.  The dual jet bundle $Z$ is a bundle over $Y$ with local coordinates $(t,q^{i},p_{t},p_{i})$ and can be identified with $T^{\ast}Y \cong T^{\ast}\mathbb{R} \times T^{\ast}Q$.  The canonical one-form on $Z$ is given by $\Theta=p_{i}dq^{i} +p_{t}dt$ and  the canonical two-form is given by $\Omega=-\mathbf{d}\theta=dq^{i} \wedge dp_{i}-dp_{t} \wedge dt$.

\subsubsection{Non-Holonomic Affine Constraints}

Consider now a distribution $\Delta$ on $Y$ which is \emph{weakly horizontal}, i.e., there exists a vertical bundle $W \subset VY$ such that $\Delta \oplus W = TY$ (see Ref.~\onlinecite{Kr2005}).  In local coordinates, let 
\begin{equation} \label{annihilator}
	\varphi^\alpha := A^\alpha_i(t, q)   dq^i +   B^\alpha(t, q) dt 
\end{equation}
with $\alpha = 1, \ldots, k$, be a basis for the annihilator $\Delta^\circ$.  The distribution $\Delta$ gives rise to an affine subbundle of $J^1 Y$, denoted by $\mathcal{C}$ and given in coordinates by 
\[
	\mathcal{C} = \{ (t, q^i, v^i) \in J^1 Y \mid
		A^\alpha_i(t, q) v^i + B^\alpha(t, q) = 0 \}.
\]
Intrinsically, $\mathcal{C}$ can be defined as follows.  Recall that a jet $\gamma \in J^1 Y$ can be viewed as a map from $T_t \mathbb{R}$ to $T_{(t, q)} Y$ which is locally of the form $\gamma = dt \otimes ( \partial/\partial t + v^i \partial/\partial q^i)$ (see \eqref{jetmap}).  With this interpretation, $\mathcal{C}$ is the set of all $\gamma \in J^1 Y$ such that $\mathrm{Im} \, \gamma \subset \Delta$.

The distribution $\Delta$ can now be lifted to the Pontryagin bundle $M=J^{1}Y \oplus Z$ as follows.  First, let $\pi_{Y, J^1 Y} : J^1 Y \to Y$ be the canonical jet bundle projection, given in coordinates by $\pi_{Y, J^1 Y}(t, q^i, v^i) = (t, q^i)$, and define now
\[
	\Delta_{J^1 Y} = T \pi_{Y, J^1 Y}^{-1}(\Delta)  \subset T(J^{1}Y).
\]
This is a distribution defined on the whole of $J^1 Y$, but only its values along $\mathcal{C}$ will be needed.  We therefore denote the restriction of this bundle to $\mathcal{C}$ by $\Delta_{\mathcal{C}}$.  If the annihilator of $\Delta$ is given by \eqref{annihilator}, then an easy coordinate computation shows that $\Delta_{\mathcal{C}}^\circ$ is spanned by the $k$ contact forms
\begin{equation}\label{OneFormZ}
\tilde{\varphi}^\alpha :=A_{i}^{\alpha}(t,q)(dq^{i} -v^{i}dt),
\end{equation}
with $\alpha = 1, \ldots, k$.  Finally, we let $\Delta_M$ be the distribution defined along points of $\mathcal{C} \times J^1 Y^\star \subset M$, given by $\Delta_M := T\pi_{J^{1}Y,M}^{-1}(\Delta_{\mathcal{C}})$, where $\pi_{J^{1}Y,M}: M \to J^{1}Y$.   The annihilator $\Delta_M^\circ$ is spanned by the pullback to $M$ of the one-forms $\tilde{\varphi}^\alpha$ in \eqref{OneFormZ}.  This is but one way of constructing $\Delta_{\mathcal{C}}$; other approaches (which, unlike our construction, can be used for nonlinear constraints as well) can be found in Refs.~\onlinecite{LeMa1996, CaLeMa1999}, and the references therein.

\subsubsection{Dirac Structures}

We now construct a Dirac structure on the Pontryagin bundle $M$ as follows: 
\[
\begin{split}
D_{M}=\left\{ (\mathcal{X}, \Sigma) \in TM \oplus T^{\ast}M  \mid  \; \Sigma-\mathbf{i}_{\mathcal{X}} \Omega_{M} \in \Delta_{M}^{\circ} \quad  \textrm{for} \quad  
\mathcal{X} \in \Delta_{M} \right\},
\end{split}
\]
where $\Omega_{M}:=\pi_{ZM}^{\ast}\Omega$.    It can be shown directly that $D_M$ is a Dirac structure (see Ref.~\onlinecite{YoMa2006a}); alternatively, note that $D_M$ is of the form described in the theorem~\ref{thm:localchar} and hence automatically determines a (multi-)Dirac structure.

Now, let  ${\mathcal{X}}: M \to TZ$ be a partial vector field as in \eqref{multvectf}.  In local coordinates, we have 
\[
	{\mathcal{X}} =
		\frac{\partial}{\partial t}
			+ \dot{q}^i \frac{\partial}{\partial q^{i}}
			+ \dot{p}_{i} \frac{\partial}{\partial p_i}
			+ \dot{p_{t}} \frac{\partial}{\partial p_{t}}.
\]
On the other hand, the generalized energy $\mathcal{E}$ is given by $\mathcal{E} = E dt$, where
\[
E(t,q^{i},v^{i},p_{t},p_{i})=p_{t}+p_{i}v^{i}-L(t,q^{i},v^{i})
\]
and the differential of $E$ is locally given by 
\[
\mathbf{d}E(t,q^{i},v^{i},p_{t},p_{i})=\left( -\frac{\partial{L}}{\partial{t}}\right)dt+\left( -\frac{\partial{L}}{\partial{q^{i}}}\right)dq^{i}+\left(p_{i}-\frac{\partial{L}}{\partial{v^{i}}}\right) dv^{i}+dp_{t}+v^{i}dp_{i}.
\]

Thus, the Lagrange-Dirac equations for nonholonomic time-dependent mechanics are given by the triple $(\mathcal{X},E, D)$ which satisfies the condition 
\begin{equation} \label{implicitnheq} 
(\mathcal{X}, \mathbf{d}E) \in D.
\end{equation}
This is in turn equivalent to the intrinsic nonholonomic Lagrange-Dirac equations:
\[
\mathbf{i}_{\mathcal{X}}\Omega-\mathbf{d}E \in \Delta_{M}^{\circ}
\quad \text{and} \quad
\mathcal{X} \in \Delta_M.
\] 
In coordinates, one obtains
\begin{equation} \label{nheqns}
 \dot{q}^{i}=v^{i}, \quad \dot{p_{i}}= \frac{\partial L}{\partial q^{i}} +\lambda_{r}A^{r}_{i},
\quad p_{i} = \frac{\partial L}{\partial v^{i}},
\end{equation}
together with the affine constraints $A^{r}_{i}(t,q)v^{i}+B^{i}(t,q)=0$ 
and the energy conservation law
\[
\frac{d}{dt}\left( p_{t}+p_{i}v^{i}-L(t,q^{i},v^{i}) \right)=0.
\]
By imposing the constraint that $E = 0$, we find that $p_{t}=L-p_{i}v^{i}$.  Note that the equations of motion \eqref{nheqns} are the analogues of the Dirac equations \eqref{implhol} for affine constraints.

One can apply a similar approach to construct multi-Dirac structures for classical field theories with nonholonomic constraints.  After constructing the distribution $\Delta_M$ in a suitable way (see e.g. Ref.~\onlinecite{VaCaLeMa2005}), the nonholonomic multi-Dirac structure will then be given again by theorem~\ref{thm:localchar} and the nonholonomic field equations will be given by \eqref{implicitnheq}. At this point, we emphasize that in the current description of nonholonomic multi-Dirac structures the constraints are assumed to be {\it linear/affine} in the multi-velocities, while $\Delta_M$ can be constructed for nonlinear constraints as well.  Most physically relevant nonholonomic field theories, however, have nonlinear constraints, and the study of nonlinear nonholonomic field theories using multi-Dirac structures will be the subject of future work.

\subsection{Elastostatics}

To conclude, we show how the {\bfi Hu-Washizu principle} from linear elastostatics (see Ref.~\onlinecite{Wa1968}) can be understood as a Hamilton-Pontryagin variational principle.  Let $\mathcal{B}$ be an open subset of $\mathbb{R}^3$ representing an elastic medium and consider the bundle $\pi: Y := \mathcal{B} \times \mathbb{R}^3 \rightarrow \mathcal{B}$.  We denote coordinates on $J^1 Y$ by $(X^I, \phi^i, F^i_I)$, where the fiber coordinates $F^i_I$ represent the {\bfi deformation gradient} of the material.  The dual bundle $Z$ has coordinates $(X^I, \phi^i, p, \pi_i^I)$, where $\pi_i^I$ will be seen to represent the first {\bfi Piola-Kirchhoff tensor}.  

Furthermore, we let $W : J^1 Y \rightarrow \mathbb{R}$ be the {\bfi stored-energy function} of the material, and we assume that $W$ is a function of $F^i_I$ only (this assumption can easily be relaxed).  We denote the  {external body} forces acting on the medium by $b_i(X)$ and we assume that there are traction boundary conditions acting on the boundary $\partial \mathcal{B}$, given by $\pi^I_i n_I = \tau_i$, where $n_I$ is the unit normal to the boundary and $\tau_i(X)$ is a prescribed traction vector.  Dirichlet boundary conditions can easily be incorporated by restricting the space of fields as in Ref.~\onlinecite{MaHu1994}, though we will not do so here.

\subsubsection{Hamilton-Pontryagin Principle}

Taking $W$ as our Lagrangian, the Hamilton-Pontryagin principle becomes 
\begin{equation} \label{HWHP1}
S(\phi, F, \pi) = \int_{\mathcal{B}} \left[\pi^I_i  \left(\frac{\partial \phi^i}{\partial X^I} - F^i_I \right) + W(F) - b_i(X) \phi^i(X) \right]  dV
 - \int_{\partial \mathcal{B}} \tau_i(X) \phi^i(X) dS,
\end{equation}
where we have incorporated the traction boundary conditions through the surface integral.  Due to the presence of the boundary integral, this action functional is more general than the Hamilton-Pontryagin principle that we dealt with in Section~\ref{sec:hpprinciple}, but for the purpose of taking variations it can be treated in exactly  the same way.  In Section~\ref{sec:HPbound}, we will formulate a different way of incorporating the boundary conditions, which can easily be generalized to the case of elastostatics on a nonlinear manifold.

Taking variations of \eqref{HWHP1} with respect to $\phi$, $F$, and $\pi$ results in 
\begin{align*}
	\delta S & =  \int_{\mathcal{B}} \left[  \left(\frac{\partial \phi^i}{\partial X^I} - F^i_I \right) \delta \pi^I_i + 
		\left(\frac{\partial W}{\partial F^i_I} - \pi^I_i \right) \delta F^i_I - \left(\frac{\partial \pi^I_i}{\partial X^I} + b_i \right) \delta \phi^i \right] dV \\
		& \qquad + \int_{\partial \mathcal{B}} \left( \pi^I_i n_I - \tau_i(X) \right) \delta \phi^i_{|\partial \mathcal{B}} dS,
\end{align*}
which leads to the following equations: 
\begin{equation} \label{fieldeqs}
	\frac{\partial \pi^I_i}{\partial X^I} + b_i = 0, \quad 
	F^i_I = \frac{\partial \phi^i}{\partial X^I}, 
	\quad \text{and} \quad \pi^I_i = \frac{\partial W}{\partial F^i_I},
\end{equation}
together with the boundary conditions 
\begin{equation} \label{boundary_conditions}
	\pi^I_i n_I = \tau_i  \quad \text{(on $\partial \mathcal{B}$)},
\end{equation}
where $n_I$ is the outward normal vector to $\partial \mathcal{B}$.	The last equation in \eqref{fieldeqs} shows that $\pi^I_i$ can be identified with the first  Piola-Kirchhoff tensor, the second equation is nothing but the definition of the deformation gradient, while the first equation is the Navier-Cauchy equation.

\subsubsection{Hu-Washizu Principle}

The variational principle \eqref{HWHP1} is not yet the Hu-Washizu variational principle.  To obtain the latter,  we work in the context of linear elasticity.  Write the elastic field $\phi^i$ as $\phi^i(x) = x^i + u^i(x)$, where $\mathbf{u} = (u^1, u^2, u^3)$ is the {\bfi displacement field}, and where material and spatial points are both denoted by $x^i$. For the stored energy function, we take a quadratic expression of the form 
\[
W(e) = \frac{1}{2}c_{ik}^{jl} e^i_j e^k_l \quad \text{with} \quad
e^i_j = \frac{1}{2} ( u^i_{, j} + u^j_{, i}),
\]
where $e^i_j$ {are} referred to as the {\bfi strain tensor}, while the {\bfi constitutive tensor} $c_{ik}^{jl}$ characterizes the elastic medium.  We now substitute these expressions into \eqref{HWHP1}.  Up to a global sign and ignoring boundary conditions and external forces for the sake of simplicity, we obtain the Hu-Washizu principle for linear elastostatics:
\[
	S(u, e, p) = \int_{\mathcal{B}} \left[ p_i^j \left( e^i_j - \frac{1}{2} ( u^i_{, j} + u^j_{, i}) \right) + W(e) \right] dV.
\]
Here, we have used the fact that the stress tensor $p_i^j=\partial{W(e)}/\partial{e^{i}_{j}}$ is symmetric in both indices.

It goes without saying that this derivation relies heavily on the Euclidian structure of the space in which the body moves.  Nevertheless, it is possible to derive a Hu-Washizu principle for \emph{nonlinear} elasticity in terms of the {\bfi Green tensor} which is valid on arbitrary Riemannian manifolds.  This modified principle can be obtained from the Hamilton-Pontryagin principle by means of symmetry reduction and will be the subject of a forthcoming paper.

\subsubsection{The Hamilton-Pontryagin Principle with Boundary Conditions}
\label{sec:HPbound}

To finish, we wish to point out that the boundary term in \eqref{HWHP1} only makes sense when the space in which the medium moves is Euclidian.  When this is not the case, the boundary conditions \eqref{boundary_conditions} can be incorporated in a different, more intrinsic way which is reminiscent of the Lagrange-d'Alembert principle in mechanics (see for instance Ref.~\onlinecite{La1962}).  We first let $S$ be the action integral \eqref{HWHP1}, but without the boundary term:
\begin{equation} \label{SHP}
S(\phi, F, \pi) = \int_{\mathcal{B}} \left[\pi^I_i  \left(\frac{\partial \phi^i}{\partial X^I} - F^i_I \right) + W(F) \right]  dV,
\end{equation}
where for the sake of convenience we have also set the external body forces $b_i(X)$ equal to zero.  We require now that the physical fields $(\phi, F, \pi)$ are those for which $S$ satisfies 
\begin{equation} \label{LDAprinciple}
	\delta S(\phi, F, \pi) = \int_{\partial \mathcal{B}} \tau_i(X) \, \delta \phi^i(X) dS,
\end{equation}
for arbitrary variations $(\delta \phi, \delta F, \delta \pi)$.  This variational principle reproduces the field equations \eqref{fieldeqs} and the boundary conditions \eqref{boundary_conditions}, but has the advantage of being closer in spirit to the Hamilton-Pontryagin principle, and of being more intrinsic: when the elastic fields take value in a manifold $Q$, the boundary term in \eqref{HWHP1} is only locally defined. By contrast, \eqref{LDAprinciple} holds globally true even when the target space is a manifold: the variations $\delta \phi$ are then elements of $TQ$, while the traction vector $\tau$ can be seen as an element of $T^\ast Q$, and the pairing between both in \eqref{LDAprinciple} is the standard duality pairing.

We now show how the field equations \eqref{fieldeqs}, together with the boundary conditions \eqref{boundary_conditions}, can be expressed intrinsically in terms of the multisymplectic form as in Theorem~\ref{thm:field_equations}.  To this end, let $\beta$ be a map from the boundary $\partial \mathcal{B}$ to the dual jet bundle $Z$, given in coordinates by 
\[
	\beta = B^I_i d \phi^i \wedge d^n X_I + B d^{n+1}X,
\]
which will encode the boundary conditions, in a form to be determined later.  We recall that in this example the Pontryagin bundle $M$ has coordinates $(X^I, \phi^i, F^i_I, p, \pi^I_i)$, and we denote a generic section of $M$ by $\psi$.

With $S$ given by \eqref{SHP}, consider now the modified Hamilton-Pontryagin principle 
\[
	\frac{d}{d\lambda} S(\psi_\lambda) \Big|_{\lambda = 0} = \int_{\partial \mathcal{B}} \psi^\ast (\mathbf{i}_{\mathcal{V}_M} \beta) 
\]
for any finite variation $\psi_\lambda$ of $\psi$, with infinitesimal variation given by the vertical vector field $\mathcal{V}_M$ (see Section~\ref{sec:hp_principle}, where this notation was introduced).  Following the proof of Proposition~\ref{prop:implEL}, but keeping the boundary terms, we see that this variational principle is equivalent to 
\[
	- \int_{\mathcal{B}} \psi^\ast ( \mathbf{i}_{\mathcal{V}_M}  \Omega_{\mathcal{E}} ) + 
		\int_{\partial \mathcal{B}} \psi^\ast (\mathbf{i}_{\mathcal{V}_M}( \Theta_M - \mathcal{E})) = 
		\int_{\partial \mathcal{B}} \psi^\ast (\mathbf{i}_{\mathcal{V}_M} \beta),
\]
leading to the field equations 
\[
	\psi^\ast ( \mathbf{i}_{\mathcal{V}_M}  \Omega_{\mathcal{E}} ) = 0
\]
and the boundary conditions
\begin{equation} \label{BCgeom}
	\psi^\ast (\mathbf{i}_{\mathcal{V}_M}( \Theta_M - \mathcal{E} - \beta)) = 0 \quad \text{(on $\partial \mathcal{B}$)},
\end{equation}
for all vertical vector fields $\mathcal{V}_M$.  To see that these expressions reduce to \eqref{fieldeqs} and \eqref{boundary_conditions}, observe that $\Theta_M$ and $\mathcal{E}$ are given in the case of elastostatics by 
\[
	\Theta_M = \pi_i^I d\phi^i \wedge d^n X_I + p d^{n+1}X,
		\quad \text{and} \quad 
	\mathcal{E} = (\pi_i^I F^i_I + p - W) d^{n + 1} X.
\]
It then follows easily that the field equations are given by \eqref{fieldeqs}.  For the boundary conditions, note that the vertical vector fields are spanned by $\{ \partial/\partial \phi^i, \partial/\partial F^i_I, \partial/\partial \pi^I_i, \partial/\partial p\}$.  When evaluated on these basis elements,  \eqref{BCgeom} vanishes unless $\mathcal{V}_M = \partial/\partial \phi^i$.  In that case, we have 
\[
	0 = \psi^\ast (\mathbf{i}_{\partial/\partial \phi^i}( \Theta_M - \mathcal{E} - \beta)) = \psi^\ast ( (\pi_i^I - B_i^I) d^n X_I ).
\]
If $\mathcal{B}$ is equipped with a Riemannian metric $g_{IJ}$, the last expression can be rewritten by noting that $\psi^\ast (d^n X_I ) = g^{-1/2} n_I dS$, where $n_I$ is the unit normal vector, $g$ is the determinant of the metric, and $dS$ is the induced volume form on $\partial \mathcal{B}$.  The boundary conditions \eqref{BCgeom} then become $\pi^I_i n_I = B^I_i n_I$, which coincides with \eqref{boundary_conditions} if we choose $\beta$ so that $B^I_i n_I = \tau_i$.


\section{Conclusions and Future Work}

In this paper, we have developed the Hamilton-Pontryagin principle for field theories and we have studied the relation between the field equations obtained from the Hamilton-Pontryagin principle and Lagrange-Dirac field theories, formulated in terms of multi-Dirac structures.  These ideas are shown to be a natural extension of Dirac structures and Lagrange-Dirac dynamical systems in mechanics.  Finally, we have given a few illustrative examples of Lagrange-Dirac field theories: the Klein-Gordon equation, Maxwell's equations, the case of time-dependent mechanics with affine nonholonomic constraints, and elastostatics.

\paragraph*{Multi-Dirac Structures for Field Theories with Boundary Conditions.}

We have shown in this paper how the Hamilton-Pontryagin principle can be augmented in order to deal with field theories with non-trivial (Neumann) boundary conditions.   We are currently working on an extended concept of multi-Dirac structure, in which the boundary conditions are automatically incorporated.  As we discuss below, we expect that our approach will be related to the \emph{Stokes-Dirac structures} of van der Schaft and Maschke (Ref.~\onlinecite{ScMa2002}) used in boundary control theory. 

\paragraph*{Space-Time Slicings and the Adjoint Formalism.}

The approach in this paper are \emph{covariant} in the sense that no distinction is made between time and the spatial variables on the base space.  In fact, as we have pointed out before, the configuration bundle $\pi_{XY}: Y \to X$ can be arbitrary and in particular $X$ does not have to represent spacetime.  While this approach has various advantages (see Ref.~\onlinecite{KiTu1979, Br1997, GoIsMa1997} for an overview), not in the least that the resulting structures are all finite-dimensional, it is often necessary to single out one coordinate direction as time, for instance when considering the initial-value problem for field theories. 

In related work (see Ref.~\onlinecite{CaVaGoMaYo2011}) we will show that when such a space-time slicing is chosen (as in Ref.~\onlinecite{GoIsMa1999, BiSnFi1988}), the Hamilton-Pontryagin variational principle yields the field equations in \emph{adjoint form}, bringing out the distinction between the \emph{dynamical} variables of the theory (the physical fields) and the \emph{kinematic} variables, which are present due to gauge freedom.   For instance, in the case of relativity, the former are the metric and the extrinsic curvature on a spatial hypersurface, the latter are the lapse and shift, and the adjoint formalism reduces to the familiar ADM equations of relativity.  

Furthermore, the choice of a space-time splitting induces a correspondence between multi-Dirac structures on the jet bundle and certain infinite-dimensional Dirac structures on the space of all fields.  The latter are in turn related to the concept of \emph{Stokes-Dirac structures} (Ref.~\onlinecite{ScMa2002}).  The precise link between Stokes-Dirac structures and infinite-dimensional Dirac structures involves Poisson reduction and was established in Ref.~\onlinecite{VaYoLeMa2010}.

\paragraph*{Discretizations of Multi-Dirac Structures and Field Theories.}

Leok and Ohsawa (Ref.~\onlinecite{LeOh2011}) recently proposed a concept of \emph{discrete Dirac structure} in mechanics.  At the same time, it has been shown in Ref.~\onlinecite{BoMa2009} that discrete versions of the Hamilton-Pontryagin principle can be used to derive very accurate geometric integrators for mechanical systems.  We propose to explore the field-theoretic version of these discretizations. Since multi-Dirac structures and variational principles are a natural way to describe field theories even for degenerate Lagrangians, this approach is especially promising.   In particular, we will focus on field theories whose covariance group includes the diffeomorphism group of the base space $X$, focusing on the link with the adjoint formalism.

\paragraph*{Symmetry and Reduction of Multi-Dirac Structures.}  

The presence of non-trivial symmetry groups is an essential feature of many classical field theories, such as electromagnetism, Yang-Mills theories, and elasticity.  More complex examples include the theory of perfect fluids and Einstein's theory of relativity.  It is often useful to obtain a form of the field equations from which the symmetry has been factored out; this is for instance what happens in writing Maxwell's equations in terms of the field strength $F_{\mu\nu}$ rather than the vector potential $A_\mu$.  From a bundle-theoretic point of view, various approaches of symmetry reduction for field theories have been proposed (see Refs.~\onlinecite{CaMa2003, CaRa2003, GaRa2010} and the references therein).  We propose to develop a reduction theory for multi-Dirac structures along the lines of the theory of symmetry reduction for  Dirac structures in classical mechanics developed by Refs.~\onlinecite{BlRa2004, YoMa2007, YoMa2009}.  

\begin{acknowledgments}
This paper could not have been written if it wasn't for the friendly guidance and support of Jerry Marsden.  We had many scientific discussions with Jerry when we visited Caltech, but unfortunately he passed away before he could see the result of our work.  We therefore would like to dedicate this paper to his memory.

We are very grateful to Frans Cantrijn, Marco Castrill{\'o}n-L{\'o}pez, Mark Gotay, Juan-Carlos Marrero, David Mart\'{\i}n de Diego, Tomoki Ohsawa, Chris Rogers and Marco Zambon for useful discussions.   In particular, we thank Marco Zambon for pointing out a flaw in an earlier version and an anonymous referee for suggestions to improve this paper.

J. V. is on leave from a Postdoctoral Fellowship of the Research Foundation--Flanders (FWO-Vlaanderen) and would like to thank JSPS for financial support during a research visit to Waseda University, where part of this work was carried out.  The research of H. Y. is partially supported by JSPS Grant-in-Aid 23560269, JST-CREST and Waseda University Grant for SR 2012A-602.  M. L. is partially supported by NSF grants DMS-1010687, and DMS-1065972.  This work is supported by the {\sc irses} project {\sc geomech} (nr.\ 246981) within the 7th European Community Framework Programme

\end{acknowledgments}


\appendix
\section{Generalities about Multi-Vector Fields} \label{sec:multivector}

\subsection{Definitions}

Let $M$ be a manifold of dimension $n$.  A $k$-multivector field (where $k \le n$) on $M$ is a section $\mathcal{X}$ of the $k$-fold exterior power $\bigwedge^k(TM)$ of the tangent bundle.  The module of all $k$-multivector fields is denoted by $\mathfrak{X}^k(M)$.

We define the {\bfi interior product} of multivector fields and forms as follows.  Let $\mathcal{X} \in \mathfrak{X}^k(M)$ and $\alpha \in \Omega^l(M)$, where $k \le l$.  The {\bfi left interior product} $\mathcal{X} \intprod \alpha$ is then the unique $(l-k)$-form such that
\[
	(\mathcal{X} \intprod \alpha)(\mathcal{X}') = \alpha(\mathcal{X} \wedge \mathcal{X}')
\]
for all $(l - k)$-multivectors $\mathcal{X}'$.  Similarly, the {\bfi right interior product} of a multivector $\mathcal{X} \in \mathfrak{X}^k(M)$ and a form $\beta \in \Omega^m(M)$, with $k \ge m$, is the unique $(k-m)$-multivector $\mathcal{X} \prodint \beta$ with the property that 
\[
	(\mathcal{X} \prodint \beta) \intprod \gamma = (\beta \wedge \gamma)(\mathcal{X})
\]
for all $(k-m)$-forms $\gamma$ on $M$.

Further information about multivectors and interior products can be found in (for instance) Ref.~\onlinecite{Tu1974, Ma1997}.  From these papers, we quote the following two results on the relation between the interior product and the wedge product of vector fields and forms, respectively.

\begin{lemma}[Ref.~\onlinecite{Tu1974}] \label{lemma:wedge}
Let $\mathcal{X} \in \mathfrak{X}^k(M)$ and $\alpha \in \Omega^1(M)$.  Then 
\[
	(\mathcal{X} \wedge \mathcal{X}') \prodint \alpha = 
		(\mathcal{X} \prodint \alpha) \wedge \mathcal{X}' + 
		(-1)^k \mathcal{X} \wedge (\mathcal{X}' \prodint \alpha)
\]
for any multivector $\mathcal{X}'$, and 
\[
	\mathcal{X} \intprod (\alpha \wedge \beta) = (\mathcal{X} \prodint \alpha) \intprod \beta
		+ (-1)^k \alpha \wedge (\mathcal{X} \intprod \beta)
\]
for any form $\beta$.
\end{lemma}

When the multi-vector field is decomposable, the right interior product with a one-form can easily be computed by means of the following formula, the proof of which follows from the previous theorem.

\begin{lemma} \label{lemma:decomp}
	Let $\mathcal{X}$ be a $k$-multi-vector field which is {\bfi decomposable} in the sense that $\mathcal{X} = \bigwedge_{\mu = 1}^k \mathcal{X}_\mu$, where $\mathcal{X}_\mu$ ($\mu = 1, \ldots, k$) are vector fields and consider a one-form $\alpha$.  Then
	\[
		\mathcal{X} \prodint \alpha = 
			\sum_{\mu = 1}^k (-1)^{\mu + 1} \left<\mathcal{X}_\mu, \alpha \right> 
				\hat{\mathcal{X}}_\mu, 
	\]
	where $\hat{\mathcal{X}}_\mu$ is the $(k-1)$-vector field obtained by deleting $\mathcal{X}_\mu$, i.e. 
	\[
		\hat{\mathcal{X}}_\mu = \bigwedge_{\stackrel{\lambda = 1}{\lambda \ne \mu}}^{n+1} \mathcal{X}_\lambda.
	\]
\end{lemma}

\begin{corollary} \label{cor:prodprod}
	Let $\mathcal{X}$ be a decomposable $k$-multivector and $\alpha$ an $l$-form, where $l \ge k$.  Then $\mathcal{X} \prodint (\mathcal{X} \intprod \alpha ) = 0$.
\end{corollary}
\begin{proof}
Let $\beta = \mathcal{X} \intprod \alpha$.  With the notations of the previous lemma, we have that 
\[
\mathcal{X} \prodint \beta = 
			\sum_{\mu = 1}^k (-1)^{\mu + 1} \left<\mathcal{X}_\mu, \beta \right> 
				\hat{\mathcal{X}}_\mu, 
	\]
but $\left<\mathcal{X}_\mu, \beta \right> = 0$ for all $\mu = 1, \ldots, k$, since
$\left<\mathcal{X}_\mu, \beta \right> = 
	\mathbf{i}_{\mathcal{X}_\mu} \mathbf{i}_{\bigwedge_{\lambda = 1}^{n+1} \mathcal{X}_\lambda} \alpha = 0$.
\end{proof}

In this paper, we also let $\hat{\mathcal{X}}_{\mu\nu}$ (where $\mu \ne \nu$) be the $(k-2)$-multivector obtained by deleting both $\mathcal{X}_\mu$ and $\mathcal{X}_\nu$: 
\[
	\hat{\mathcal{X}}_{\mu\nu} = \bigwedge_{\stackrel{\lambda = 1}{\lambda \ne \mu, \nu}}^{n+1} \mathcal{X}_\lambda.
\]
If $\mu < \nu$, we then have that 
\[
	\mathcal{X}_\lambda \wedge \hat{\mathcal{X}}_{\mu\nu} 
		= (-1)^{\mu+1} \delta_{\mu\lambda} \hat{\mathcal{X}}_\nu 
		-  (-1)^{\nu+1} \delta_{\nu\lambda} \hat{\mathcal{X}}_\mu.
\]
If $\mu > \nu$, note that $\hat{\mathcal{X}}_{\mu\nu} = \hat{\mathcal{X}}_{\nu\mu}$.  The above formula can then be applied.

\subsection{Proof of \eqref{contract}}
\label{sec:comp}

The computation of equation \eqref{contract} is somewhat involved and depends on a number of coordinate identities, listed here.  First, recall that $\mathcal{X}$ is an $(n + 1)$-multivector field with local expression \eqref{multvectf}, and that the multi-symplectic form $\Omega$ is locally given by equation \eqref{msform}.

Recall the volume form $\eta$ on $X$ is locally expressed by $d^{n + 1} x$ and denote
\[
	d^n x_\mu := \frac{\partial}{\partial x^\mu} \intprod d^{n+1}x, \quad
	d^{n-1}x_{\mu\nu} := \frac{\partial}{\partial x^\nu} \intprod d^n x_\mu, \quad \ldots
\]	
The following contractions will be useful:
\begin{equation} \label{contr1}
	\hat{\mathcal{X}}_\mu \intprod d^n x_\nu  = (-1)^{\nu + n} \delta_{\mu\nu}
\end{equation}
and, for $\mu < \nu$,
\begin{equation} \label{contr2}
	\hat{\mathcal{X}}_{\mu\nu} \intprod d^n x_\lambda = 
		(-1)^{n + \mu + \nu} ( \delta_{\nu\lambda} dx^\mu + \delta_{\mu\lambda} dx^\nu).
\end{equation}

Now return to the calculation of equation \eqref{contract} and the contraction of $\mathcal{X}$ with $\Omega_{M}$ is now given by 
\[
	\mathcal{X} \intprod \Omega_{M} = 
		\mathcal{X} \intprod (dy^A \wedge dp_A^\mu \wedge d^n x_\mu) -
		\mathcal{X} \intprod (dp \wedge d^{n+1} x), 
\]
and the two terms on the right-hand side will be calculated separately.

We begin with the first term. Using lemma~\ref{lemma:wedge}, we have 
\[
\mathcal{X} \intprod (dy^A \wedge dp_A^\mu \wedge d^n x_\mu)
= (\mathcal{X} \prodint dy^A) \intprod (dp_A^\mu \wedge d^n x_\mu) + 
(-1)^{n+1} dy^A \wedge ( \mathcal{X} \intprod (dp_A^\mu \wedge d^n x_\mu)).
\]
Both terms can be calculated using lemma~\ref{lemma:decomp} and using \eqref{contr1} and \eqref{contr2}.  After some rearrangements, the result is that 
\[
	\mathcal{X} \intprod (dy^A \wedge dp_A^\mu \wedge d^n x_\mu) = 
	(-1)^{n+2} \left[ (C^A_\mu C_{A\lambda}^\lambda - C^A_\lambda C^\lambda_{A\mu}) dx^\mu + C^A_\mu dp_A^\mu \right].
\]

For the second term, we use the same techniques to conclude that 
\[
	\mathcal{X} \intprod (dp \wedge d^{n+1} x) = (-1)^{n+1} \left[ dp - C_\mu dx^\mu \right].
\]
Putting the results for both terms together, we obtain \eqref{contract}.


\providecommand{\arxiv}[1]{\texttt{arXiv:#1}}\providecommand{\href}[1]{\texttt%
{#1}}\def\polhk#1{\setbox0=\hbox{#1}{\ooalign{\hidewidth\lower1.5ex\hbox{`}\hi%
dewidth\crcr\unhbox0}}} \def\cprime{$'$}

\end{document}